\documentclass{article}
\pdfoutput=1

\usepackage{microtype}
\usepackage{graphicx}
\usepackage{subfigure}
\usepackage{booktabs} 
\usepackage{amsmath}
\usepackage{amsthm}
\usepackage{caption}
\usepackage{amssymb}
\usepackage{xspace}
\usepackage{algorithm}
\usepackage{mathtools}
\usepackage{forest}
\usepackage[noend]{algorithmic}
\usepackage[hyphens]{url}
\usepackage[T1]{fontenc}
\usepackage{bm}
\usepackage[short]{optidef}

\usepackage[accepted]{icml2020}
\usepackage{hyperref}
\definecolor{mylinkcolor}{RGB}{0,0,140}
\hypersetup{colorlinks,allcolors=mylinkcolor,citecolor=mylinkcolor}
\usepackage[capitalize]{cleveref}

\newcommand{\promo}{\textsc{Promotion}\xspace}
\newcommand{\agree}{\textsc{Agreement}\xspace}
\newcommand{\disagree}{\textsc{Disagreement}\xspace}
\newcommand{\subsetsum}{\textsc{Subset Sum}\xspace}
\newcommand{\sfwork}{\textsc{SFWork}\xspace}

\newcommand{\allstate}{\textsc{Allstate}\xspace}
\newcommand{\yoochoose}{\textsc{Yoochoose}\xspace}

\newtheorem{theorem}{Theorem}
\newtheorem{proposition}{Proposition}
\newtheorem{corollary}{Corollary}
\newtheorem{lemma}{Lemma}
\theoremstyle{definition}
\newtheorem{definition}{Definition}

\newcommand{\hide}[1]{}

\newcommand{\xhdr}[1]{\vspace{0.0mm}\noindent{{\bf #1.}}\hspace{0.5mm}}

\DeclareMathOperator*{\argmax}{arg\,max}

\begin{document}

\twocolumn[
\icmltitle{Choice Set Optimization Under Discrete Choice Models of Group Decisions}

\icmltitlerunning{Choice Set Optimization Under Discrete Choice Models of Group Decisions}



\icmlsetsymbol{equal}{*}

\begin{icmlauthorlist}
\icmlauthor{Kiran Tomlinson}{cornell}
\icmlauthor{Austin R.~Benson}{cornell}
\end{icmlauthorlist}

\icmlaffiliation{cornell}{Department of Computer Science, Cornell University, Ithaca, New York, USA}

\icmlcorrespondingauthor{Kiran Tomlinson}{kt@cs.cornell.edu}
\icmlcorrespondingauthor{Austin R.~Benson}{arb@cs.cornell.edu}

\icmlkeywords{discrete choice, disagreement, group decision-making}

\vskip 0.3in
]



\printAffiliationsAndNotice{}  

\begin{abstract}
The way that people make choices or exhibit preferences can be strongly
affected by the set of available alternatives, often called the choice set.
Furthermore, there are usually heterogeneous preferences, either at an individual level
within small groups or within sub-populations of large groups.
Given the availability of choice data, there are now many models that capture this behavior
in order to make effective predictions---however, there is little work in understanding how directly changing
the choice set can be used to influence the preferences of a collection of decision-makers.
Here, we use discrete choice modeling to develop an optimization framework of such interventions for
several problems of group influence, namely
maximizing agreement or disagreement and promoting a particular choice.
We show that these problems are NP-hard in general, but imposing
restrictions reveals a fundamental boundary: 
promoting a choice can be easier than encouraging consensus or sowing discord.
We design approximation algorithms for the hard problems and show
that they work well on real-world choice data.
\end{abstract}

\section{Context effects and optimizing choice sets}
Choosing from a set of alternatives is one of the most important actions people take,
and choices determine the composition of governments, the success of corporations, and the formation of social connections.
For these reasons, choice models have received significant attention in the fields of 
economics~\cite{train2009discrete}, 
psychology~\cite{tversky1981framing}, and, 
as human-generated data has become increasingly available online, computer science~\cite{overgoor2019choosing,seshadri2019discovering,rosenfeld2019predicting}.
In many cases, it is important that people have heterogeneous preferences; 
for example, people living in different parts of a town might prefer different government policies.

Much of the computational work on choice has been devoted to fitting models for predicting future choices.
In addition to prediction, another area of interest is determining effective interventions to influence 
choice---advertising and political campaigning are prime examples.
In heterogeneous groups, the goal might be to encourage consensus~\cite{amir2015care}, or,
for an ill-intentioned adversary, to sow discord, e.g., amongst political parties~\cite{rosenberg2020chaos}.

One particular method of influence is introducing new alternatives or options.
While early economic models assume that alternatives are irrelevant to the relative ranking of options~\cite{luce1959individual,mcfadden1974conditional},
experimental work has consistently found that new alternatives have strong effects on our choices~\cite{huber1982adding,simonson1992choice,shafir1993reason,trueblood2013context}.
These effects are often called \emph{context effects} or \emph{choice set effects}.
A well-known example is the compromise effect~\cite{simonson1989choice},
which describes how people often prefer a middle ground (e.g., the middle-priced wine).
Direct measurements on choice data have also revealed choice set effects in several
domains~\cite{benson2016relevance,seshadri2019discovering}.

Here, we pose adding new alternatives as a discrete optimization problem for influencing a collection of decision makers, such as the inhabitants of a city or the visitors to a website.
To this end, we consider various models for how someone makes a choice from a given set of alternatives,
where the model parameters can be readily estimated from data.
In our setup, everyone has a base set of alternatives from which they make a choice,
and the goal is to find a set of additional alternatives to optimize some function of the group's joint preferences on the base set.
We specifically analyze three objectives:
(i) \emph{agreement} in preferences amongst the group;
(ii) \emph{disagreement} in preferences amongst the group; and
(iii) \emph{promotion} of a particular item (decision).

We use the framework of \emph{discrete choice}~\cite{train2009discrete} to probabilistically model a person's choice from a given set of items, called the \emph{choice set}.
These models are parameterized for individual preferences, and
when fitting parameters from data, preferences are commonly aggregated at the level of a sub-population of individuals.
Discrete choice models such as the multinomial logit and elimination-by-aspects have played a central role in behavioral economics for several decades with diverse applications, including
forest management~\cite{hanley1998using}, social networks formation~\cite{overgoor2019choosing}, and marketing campaigns~\cite{fader1990elimination}.
More recently, new choice data and algorithms have spurred machine learning research
on models for choice set effects~\cite{ragain2016pairwise,chierichetti2018learning,seshadri2019discovering,pfannschmidt2019learning,rosenfeld2019predicting,bower2020salient}.

We provide the relevant background on discrete choice models in \cref{sec:models}.
From this, we formally define and analyze three choice set optimization problems---\agree, \disagree, and \promo---and
analyze them under four discrete choice models:
multinomial logit~\cite{mcfadden1974conditional},
the context dependent random utility model~\cite{seshadri2019discovering},
nested logit~\cite{mcfadden1978modeling}, and
elimination-by-aspects~\cite{tversky1972elimination}.
We first prove that the choice set optimization problems are NP-hard in general for these models.
After, we identify natural restrictions of the problems under which they become tractable.
These restrictions reveal a fundamental boundary:
promoting a particular item within a group is easier than minimizing or maximizing consensus.
More specifically, we show that restricting the choice models can make \promo tractable while leaving \agree and \disagree NP-hard, 
indicating that the interaction between individuals introduces significant complexity to choice set optimization.

After this, we provide efficient approximation algorithms with guarantees for all three problems under several choice models,
and we validate our algorithms on choice data.
Model parameters are learned for different types of individuals based on features (e.g., where someone lives).
From these learned models, we apply our algorithms to optimize group-level preferences.
Our algorithms outperform a natural baseline on real-world
data coming from transportation choices, insurance policy purchases, and online shopping.

\subsection{Related work}

Our work fits within recent interest from computer science and machine learning on discrete choice models in general
and choice set effects in particular.
For example, choice set effects abundant in online data has led to richer data models~\cite{ieong2012predicting,chen2016predicting,ragain2016pairwise,seshadri2019discovering,makhijani2019parametric,rosenfeld2019predicting,bower2020salient},
new methods for testing the presence of choice set effects~\cite{benson2016relevance,seshadri2019discovering,seshadri2019fundamental}, and
new learning algorithms~\cite{kleinberg2017comparison,chierichetti2018learning}.
More broadly, there are efforts on learning algorithms for
multinomial logit mixtures~\cite{oh2014learning,ammar2014s,kallus2016revealed,zhao2019learning},
Plackett-Luce models~\cite{maystre2015fast,zhao2016learning}, and
other random utility models~\cite{oh2015collaboratively,chierichetti2018discrete,benson2018discrete}.

One of our optimization problems is maximizing group agreement by introducing new alternatives.
This is motivated in part by how additional context can sway
opinion on controversial topics~\cite{munson2013encouraging,liao2014can,graells2014people}.
There are also related algorithms for decreasing polarization in
social networks~\cite{garimella2017reducing,matakos2017measuring,chen2018quantifying,musco2018minimizing},
although we have no explicit network and adopt a choice-theoretic framework.

Our choice set optimization framework is similar to assortment optimization in operations research, where the goal is find the optimal set of products to offer in order to maximize revenue~\cite{talluri2004revenue}. Discrete choice models are extensively used in this line of research, including the multinomial logit~\cite{rusmevichientong2010dynamic,rusmevichientong2014assortment} and nested logit~\cite{gallego2014constrained,davis2014assortment} models. We instead focus our attention primarily on optimizing agreement among individuals, which has not been explored in traditional revenue-focused assortment optimization.

Finally, our problems relate to group decision-making.
In psychology, introducing new shared information is critical for group decisions~\cite{stasser1985pooling,lu2012twenty}.
In computer science, the complexity of group Bayesian reasoning
is a concern~\cite{hkazla2017bayesian,hazla2019reasoning}.

\section{Background and preliminaries}\label{sec:models}
We first introduce the discrete choice models that we analyze.
In the setting we explore, one or more individuals make a (possibly random) \emph{choice} 
of a single item (or alternative) from a finite set of items called a \emph{choice set}.
We use $\mathcal U$ to denote the universe of items and $C\subseteq \mathcal U$ the choice set.
Thus, given $C$, an individual chooses some item $x \in C$.

Given $C$, a discrete choice \emph{model} provides a probability for choosing each item $x \in C$.
We analyze four broad discrete choice models that are all \emph{random utility models} (RUMs), which derive from economic rationality.
In a RUM, an individual observes a random utility for each item $x \in C$ and then chooses the one with the largest utility.
We model each individual's choices through the same RUM but with possibly different parameters to capture preference heterogeneity.
In this sense, we have a mixture model.

Choice data typically contains many observations from various choice sets.
We occasionally have data specific enough to model the choices of a particular individual, but often only one choice is recorded per person, making accurate preference learning impossible at that scale. 
Thus, we instead model the heterogeneous preferences of sub-populations or categories of individuals. 
For convenience, we still use ``individual'' or ``person'' when referring to components of a mixed population, since we can treat each component as a decision-making agent with its own preferences.
In contrast, we use the term ``group'' to refer to the entire population.
We use $A$ to denote the set of individuals (in the broad sense above), and $a \in A$ indexes model parameters.

The parameters of the RUMs we analyze can be inferred from data, and
our theoretical results and algorithms assume that we have learned these parameters.
Our analysis focuses on how the probability of selecting an item $x$ from a choice set $C$ changes as 
we add new alternative items from $\overline C = \mathcal U \setminus C$ to the choice set.

We let $n = |A|$, $k = |C|$, and $m = |\overline C|$ for notation.
We mostly use $n = 2$, which is sufficient for hardness proofs.


\xhdr{Multinomial logit (MNL)}
The multinomial logit (MNL) model~\cite{mcfadden1974conditional} is the workhorse of discrete choice theory.
In MNL, an individual $a$'s preferences are encoded by a true utility $u_a(x)$ for every item $x \in \mathcal U$.
The observations are noisy random utilities $\tilde{u}_a(x) = u_a(x) + \varepsilon$, where $\varepsilon$ follows a Gumbel distribution. 
Under this model, the probability that individual $a$ picks item $x$ from choice set $C$ (i.e., $x = \argmax_{y \in C} \tilde{u}_a(y)$)
is the softmax over item utilities:
\begin{equation}
   \Pr(a\gets x \mid C) = \frac{e^{u_a(x)}}{\sum_{y \in C}e^{u_a(y)}}.    
\end{equation}  
We use the term \emph{exp-utility} for terms like $e^{u_a(x)}$.
The utility of an item is often parameterized as a function of features of the item in order to generalize to unseen data.
For example, a linear function is an additive utility model~\cite{tversky1993context} and looks like logistic regression.
In our analysis, we work directly with the utilities.

The MNL satisfies \emph{independence of irrelevant alternatives} (IIA)~\cite{luce1959individual}, 
the property that for any two choice sets $C, D$ and two items $x, y\in C\cap D$:
$\frac{\Pr(a\gets x \mid C)}{\Pr(a\gets y \mid C)} = \frac{\Pr(a\gets x \mid D)}{\Pr(a\gets y \mid D)}$.
In other words, the choice set has no effect on $a$'s relative probability of choosing $x$ or $y$.%
\footnote{Over $a \in A$, we have a mixed logit which does
not have to satisfy IIA~\cite{mcfadden2000mixed}. Here, we are interested in the IIA property at the individual level.}
Although IIA is intuitively pleasing, behavioral experiments show that it is often violated in practice~\cite{huber1982adding,simonson1992choice}. 
Thus, there are many models that account for IIA violations, including the other ones we analyze.  

\xhdr{Context-dependent random utility model (CDM)}
The CDM~\cite{seshadri2019discovering} is an extension of MNL that can model IIA violations.
The core idea is to approximate choice set effects by the effect of each item's presence on the utilities of the other items. 
For instance, a diner's preference for a ribeye steak may decrease relative to a fish option if filet mignon is also available.
Formally, each item $z$ exerts a pull on $a$'s utility from $x$, which we denote $p_a(z, x)$.
The CDM then resembles the MNL with utilities $u_a(x\mid C)= u_a(x) + \sum_{z\in C} p_a(z, x)$. 
This leads to choice probabilities that are a softmax over the context-dependent utilities:
\begin{equation}\label{eq:cdm}
  \Pr(a\gets x \mid C) = \frac{e^{u_a(x\mid C)}}{\sum_{y \in C}e^{u_a(y\mid C)}}.    
\end{equation}

\xhdr{Nested logit (NL)}
The nested logit (NL) model~\cite{mcfadden1978modeling} instead accounts for choice set effects by grouping similar items into \emph{nests} that people choose between successively. For example, a diner may first choose between a vegetarian, fish, or steak meal and then select a particular dish.
NL can be derived by introducing correlation between the random utility noise $\varepsilon$ in MNL; 
here, we instead consider a generalized tree-based version of the model.%
\footnote{Certain parameter regimes in this generalized model do not correspond to RUMs~\cite{train2009discrete},
but this model is easier to analyze and captures the salient structure.}

The (generalized) NL model for an individual $a$ consists of a tree $T_a$ with a leaf for each item in $\mathcal U$,
where the internal nodes represent categories of items. Rather than having a utility only on items, each person $a$ also has utilities $u_a(v)$ on all nodes $v\in T_a$ (except the root). Given a choice set $C$, let $T_a(C)$ be the subtree of $T_a$ induced by $C$ and all ancestors of $C$. To choose an item from $C$, $a$ starts at the root and repeatedly picks between the children of the current node according to the MNL model until reaching a leaf.

\xhdr{Elimination-by-aspects (EBA)}
While the previous models are based on MNL, the elimination-by-aspects (EBA) model~\cite{tversky1972elimination} has a different structure. 
In EBA, each item $x$ has a set of \emph{aspects} $x'$ representing properties of the item, and person $a$
has a utility $u_a(\chi) > 0$ on each aspect $\chi$. An item is chosen by repeatedly picking an aspect with probability proportional to its utility and eliminating all items that do not have that aspect until only one item remains (or, if all remaining items have the same aspects, the choice is made uniformly at random). For example, a diner may first eliminate items that are too expensive, then disregard meat options, and finally look for dishes with pasta before choosing mushroom ravioli.
Formally, let $C' = \bigcup_{x\in C} x'$ be the set of aspects of items in $C$ and let $C^0 = \bigcap_{x\in C} x'$ be the aspects shared by all items in $C$. Additionally, let $C_\chi = \{x \in C \mid \chi \in x'\}$. The probability that individual $a$ picks item $x$ from choice set $C$ is recursively defined as
\begin{equation}\label{eq:eba}
 \Pr(a \gets x \mid C) = \frac{\sum_{\chi \in x' \setminus C^0} u_a(\chi)\Pr(a \gets x \mid C_\chi)}{\sum_{\psi \in C'\setminus C^0} u_a(\psi)}.
\end{equation}
If all remaining items have the same aspects ($C' = C^0$), the denominator is zero, and $\Pr(a \gets x \mid C) = \frac{1}{|C|}$ in that case.

\xhdr{Encoding MNLs in other models}
Although the three models with context effects appear quite different, they all subsume the MNL model.
Thus, if we prove a problem hard under MNL, then it is hard under all four models.

\begin{lemma}\label{lemma:mnl_special_case}
The MNL model is a special case of the CDM, NL, and EBA models.
\end{lemma}

\begin{proof}
Let $\mathcal M$ be an MNL model. 
For the CDM, use the utilities from $\mathcal M$ and set all pulls to 0. 
For NL, make all items children of $T_a$'s root and use the utilities from $\mathcal M$. 
Lastly, for EBA,
assign a unique aspect $\chi_x$ to each item $x \in \mathcal U$ with utility $u_a(\chi_x) = e^{u_a(x)}$. 
Following \eqref{eq:eba},
\[\Pr(a \gets x \mid C) = \frac{u_a(\chi_x)\Pr(a \gets x \mid C_{\chi_x})}{\sum_{\psi \in C'\setminus C^0} u_a(\psi)}.\]
Since $C_{\chi_x} = \{x\}$, $\Pr(a \gets x \mid C_{\chi_x}) = 1$ and thus $\Pr(a \gets x \mid C) \propto u_a(\chi_x) = e^{u_a(x)}$, matching the MNL $\mathcal M$.
\end{proof}

\section{Choice set optimization problems}
By introducing new alternatives to the choice set $C$, we can modify the relationships amongst individual preferences, resulting in different dynamics at the collective level.
Similar ideas are well-studied in voting models, e.g., introducing alternatives to change winners selected by Borda count~\cite{easley2010networks}.
Here, we study how to optimize choice sets for various group-level objectives, measured in terms of individual
choice probabilities coming from discrete choice models.

\xhdr{Agreement and Disagreement}
Since we are modeling the preferences of a collection of decision-makers, one important metric is the amount of disagreement (conversely, agreement) about which item to select. Given a set of alternatives $Z\subseteq \overline C$ we might introduce, we quantify the disagreement this would induce as the sum of all pairwise differences between individual choice probabilities over $C$:
\begin{align}
 D(Z) = \smashoperator{\sum_{\{a, b\}\subseteq A, x\in C}} |\Pr(a\gets x\mid C\cup Z)-\Pr(b\gets x\mid C\cup Z)|. \label{eq:D}
\end{align}
Here, we care about the disagreement on the original choice set $C$ that results from preferences over the new choice set $C \cup Z$.
In this setup, $C$ could represent core options (e.g., two major health care policies under deliberation)
and $Z$ additional alternatives designed to sway opinions.

Concretely, we study the following problem: 
given $A, C, \overline C$, and a choice model, minimize (or maximize) $D(Z)$ over $Z\subseteq \overline C$.
We call the minimization problem \agree and the maximization problem \disagree. \agree has applications in encouraging consensus, 
while \disagree yields insight into how susceptible a group may be to an adversary who wishes to increase conflict. Another potential application for \disagree is to enrich the diversity of preferences present in a group.  

\xhdr{Promotion}
Promoting an item is another natural objective, which is of considerable interest in online advertising and content recommendation. Given $A, C, \overline C$, a choice model, and a target item $x^*\in C$, the \promo problem is to find the set of alternatives $Z\subseteq \overline C$ whose introduction maximizes the number of individuals whose ``favorite'' item in $C$ is $x^*$.
Formally, this means maximizing the number of individuals $a \in A$ for whom
$\Pr(a \gets x^* \mid C\cup Z) > \Pr(a \gets x\mid C \cup Z)$, $x \in C$, $x \neq x^*$.
This also has applications in voting, where questions about the influence of new candidates constantly arise.

One of our contributions in this paper is showing that promotion can be easier (in a computational complexity sense)
than agreement or disagreement optimization.

\section{Hardness results}
We now characterize the computational complexity of \agree, \disagree, and \promo under the four discrete choice models. 
We first show that \agree and \disagree are NP-hard under all four models and that \promo is NP-hard under the three models with context effects.
After, we prove that imposing additional restrictions on these discrete choice models can make \promo tractable while leaving \agree and \disagree NP-hard.
The parameters of some choice models have extra degrees of freedom, e.g., MNL has additive-shift-invariant utilities.
For inference, we use a standard form (e.g., sum of utilities equals zero).
For ease of analysis, we do not use such standard forms,
but the choice probabilities remain unambiguous.

\subsection{\agree}
Although the MNL model does not have any context effects, introducing alternatives to the choice set can still affect the relative preferences of two different individuals. In particular, introducing alternatives can impact disagreement in a sufficiently complex way to make identifying the optimal set of alternatives computationally hard.
Our proof of \cref{thm:mnl_agree_hard} uses a very simple MNL in the reduction, with only two individuals and two items in $C$, where
the two individuals have \emph{exactly the same utilities on alternatives}.
In other words, even when individuals agree about new alternatives, encouraging them to agree over the choice set is hard. 

\begin{theorem}\label{thm:mnl_agree_hard}
  In the MNL model, \agree is NP-hard,
  even with just two items in $C$
  and two individuals that have identical utilities on items in $\overline C$.
\end{theorem}
\begin{proof}
By reduction from \textsc{Partition}, an NP-complete problem~\cite{karp1972reducibility}. Let $S$ be the set of integers we wish to partition into two subsets with equal sum. We construct an instance of \disagree with $A=\{a, b\}$, $C = \{x, y\}$, $\overline C=S$ (abusing notation to identify alternatives with the \textsc{Partition} integers). Let $t = \frac{1}{2}\sum_{z\in S} z$. 
Define the utilities as: 
$u_{a}(x) = \log t$, $u_{b}(x) = \log 3t$,
$u_{a}(y) = \log t$, $u_{b}(y) = \log 2t$,
and $u_a(z) = u_b(z) = \log z$ for all $z \in \overline C$.
%
The disagreement induced by a set of alternatives $Z\subseteq \overline C$ is characterized by its sum of exp-utility $s_Z = \sum_{z \in Z} z$:
\[
 D(Z) = \left|\frac{t}{2t + s_Z} - \frac{3t}{5t+ s_Z}\right| + \left|\frac{t}{2t + s_Z} - \frac{2t}{5t+s_Z}\right|.
\]
The total exp-utility of all items in $\overline C$ is $2t$.
On the interval $[0, 2t]$, $D(Z)$ is minimized at $s_Z = t$ (\cref{fig:D_plots}, left).
Thus, if we could efficiently find the set $Z$ minimizing $D(Z)$, then we could efficiently solve \textsc{Partition}.
\end{proof}

From \cref{lemma:mnl_special_case}, the other models we consider can all encode any MNL instance, which leads to the following corollary.
\begin{corollary}\label{cor:agree_hard}
  \agree is NP-hard in the CDM, NL, and EBA models.
\end{corollary}

\begin{figure}[t]
    \centering
    \includegraphics[width=0.49\columnwidth]{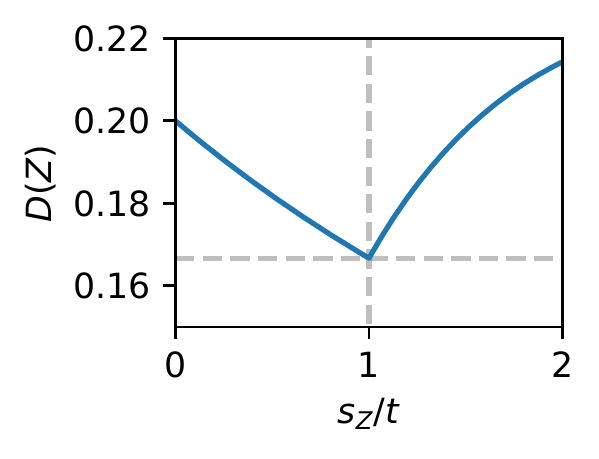}
    \hfill
    \includegraphics[width=0.49\columnwidth]{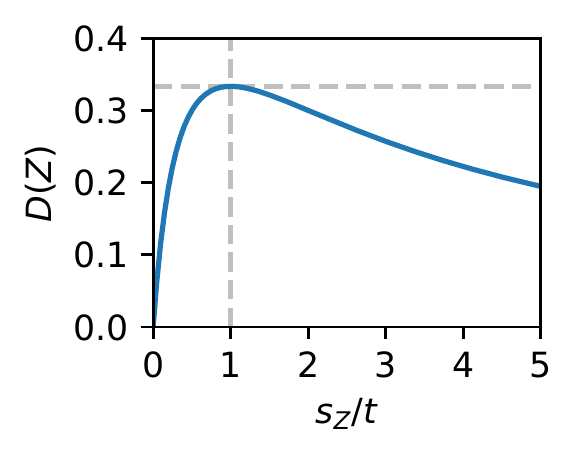}
    \vspace{-5mm}
    \caption{(Left) Plot of $\textstyle D(Z) = \big|\frac{t}{2t + s_Z} - \frac{3t}{5t+ s_Z}\big| + \big|\frac{t}{2t + s_Z} - \frac{2t}{5t+s_Z}\big|$ from the proof of \Cref{thm:mnl_agree_hard}. (Right) Plot of $D(Z)=\big| \frac{2t}{2t + s_Z} - \frac{ t/2}{t/2 + s_Z} \big|$ from the proof of \Cref{thm:mnl_disagree_hard}. Both functions are re-parameterized in terms of the ratio $s_Z/t$ by dividing through by $t$ and achieve local optima at $s_Z/t=1$ (i.e.~$s_Z=t$); this can be verified analytically.}
    \label{fig:D_plots}
\end{figure}

\subsection{\disagree}
Using a similar strategy, we can construct an MNL instance whose disagreement is maximized rather than minimized at a particular target value (\cref{thm:mnl_disagree_hard}).
The reduction requires an even simpler MNL setup.

\begin{theorem}\label{thm:mnl_disagree_hard}
  In the MNL model, \disagree is NP-hard,
  even with just one item in $C$
  and two individuals that have identical utilities on items in $\overline C$.
\end{theorem}
\begin{proof}
By reduction from \textsc{Subset Sum}~\cite{karp1972reducibility}. Let $S$ be a set of positive integers with target $t$. Let $A=\{a, b\}$, $C = \{x\}$, $\overline C=S$, with utilities:
$u_{a}(x) =  \log 2t$, $u_{b}(x) =  \log t/2$, and
$u_a(z) = u_b(z) = \log z$ for all $z \in \overline C$.
Letting $s_Z = \sum_{z \in Z} z$, including $Z\subseteq \overline C$ makes the disagreement
\begin{align*}
D(Z)=\left| \frac{2t}{2t + s_Z} - \frac{ t/2}{t/2 + s_Z} \right|.
\end{align*}
For $s_Z \ge 0$, $D(Z)$ is maximized at $s_Z = t$ (\cref{fig:D_plots}, right). Thus, if we could efficiently maximize $D(Z)$, 
then we could efficiently solve \textsc{Subset Sum}.
\end{proof}

By \cref{lemma:mnl_special_case}, we again have the following corollary. 
\begin{corollary}
  \disagree is NP-hard in the CDM, NL, and EBA models.
\end{corollary}

\subsection{\promo}
In choice models with no context effects, \promo has a constant-time solution: under IIA, the presence of alternatives has no effect on an individual's relative preference for items in $C$.
However, \promo is more interesting with context effects, and we show that it is NP-hard for CDM, NL, and EBA.
In \Cref{sec:restricted}, we will show that restrictions of these models make \promo tractable
but keep \agree and \disagree hard.

\begin{theorem}
\label{thm:cdm_promo_hard}
  In the CDM model, \promo is NP-hard,
  even with just one individual and three items in $C$.
\end{theorem}
\begin{proof}
By reduction from \textsc{Subset Sum}. 
Let set $S$ with target $t$ be an instance of \textsc{Subset Sum}. Let $A=\{a\}$, $C = \{x^*, w, y\}$, $\overline C=S$.
Using tuples interpreted entry-wise for brevity, suppose that we have the following utilities:
\begin{align*}
	u_{a}(\langle x^*, w, y \rangle \mid C) &=  \langle 1, t, -t \rangle \\
	u_a(z) &= -\infty &&\forall z \in \overline C\\
	p_{a}(z, \langle x^*, w, y\rangle ) &=  \langle z, 0, 2z \rangle  &&\forall z \in \overline C.
\end{align*}
We wish to promote $x^*$. Let $s_Z =\sum_{z\in Z} z$. 
When we include the alternatives in $Z$, $x^*$ is the item in $C$ most likely to be chosen if and only if $1+s_Z > t$ and $1+s_Z > -t+2s_Z$. 
Since $s_Z$ and $t$ are integers, this is only possible if $s_Z=t$. 
Thus, if we could efficiently promote $x^*$, then we could efficiently solve \textsc{Subset Sum}.
\end{proof}

We use the same Goldilocks strategy in our proofs for the NL and EBA models: by carefully defining utilities, we create choice instances where the optimal promotion solution is to pick just the right quantity of alternatives to increase preference for one item without overshooting.
However, the NL model has a novel challenge compared to the CDM. 
With CDM, alternatives can increase the choice probability of an item in $C$, but
in the NL, new alternatives only lower choice probabilities.
\begin{theorem}\label{thm:nested_logit_promo_hard}
In the NL model, \promo is NP-hard, even with just two individuals
and two items in $C$.
\end{theorem}

\begin{proof}
  By reduction from \textsc{Subset Sum}.
  Let $S=\{z_1, \dots, z_n\}, t$ be an instance of \textsc{Subset Sum}.
  Let $A=\{a, b\}$, $C = \{x^*, y\}$, $\overline C=S$, and $0 < \varepsilon< 1$. The nest structures and utilities are 
  shown in \cref{fig:logit_trees}.
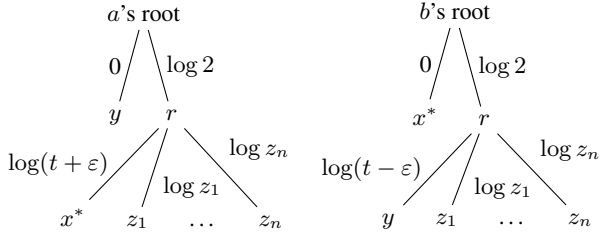
\begin{figure}[h]
\begin{footnotesize}
\begin{forest}
for tree={
  l sep=3em,
  s sep=1em,
},
[$a$'s root
  [$y$, edge label={node[midway, left] {0}}]
  [$r$, edge label={node[midway, right] {$\log 2$}}
    [$x^*$, edge label={node[midway, left] {$\log(t+\varepsilon)\;$}}]
    [$z_1$, edge label={node[midway, below right] {$\log z_1$}}]
    [\dots, no edge]
    [$z_n$, edge label={node[midway, above right] {$\log z_n$}}]
  ] 
]
\end{forest}
\,
\begin{forest}
for tree={
  l sep=3em,
  s sep=1em,
},
[$b$'s root
  [$x^*$, edge label={node[midway, left] {0}}]
  [$r$, edge label={node[midway, right] {$\log 2$}}
    [$y$, edge label={node[midway, left] {$\log(t-\varepsilon)\;$}}]
    [$z_1$, edge label={node[midway, below right] {$\log z_1$}}]
    [\dots, no edge]
    [$z_n$, edge label={node[midway, above right] {$\log z_n$}}]
  ] 
]
\end{forest}
\end{footnotesize}
\caption{NL trees used in the proof of \cref{thm:nested_logit_promo_hard}. The left tree is for individual $a$ and the right tree for individual $b$.}
\label{fig:logit_trees}
\end{figure}
 We wish to promote $x^*$. With just the choice set $C$, $a$ prefers $x^*$ to $y$, but $b$ does not. To make $b$ prefer $x^*$ to $y$, we need to cannibalize $y$ by adding $z_i$ items. However, this simultaneously cannibalizes $x^*$ in $a$'s tree, so we need to be careful not to introduce too much additional utility. To ensure $a$ prefers $x^*$, we need to pick $Z$ such that
\begin{align*}
\Pr(a\gets y \mid C\cup Z) &<  \Pr(a\gets y \mid C\cup Z)\\
\iff \frac{1}{1 + e^{\log 2}} &< \frac{e^{\log 2}}{1 + e^{\log 2}} \cdot \frac{e^{\log(t+\varepsilon)}}{e^{\log(t+\varepsilon)} + \sum_{z \in Z} e^{\log z}}\\
\iff \frac{1}{3} &< \frac{2}{3} \cdot \frac{t+\varepsilon}{t+\varepsilon + \sum_{z \in Z} z}\\
\iff \sum_{z \in Z} z &< t+\varepsilon.
\end{align*}
To ensure $b$ prefers $x^*$, we need 
\begin{align*}
\Pr(b\gets x^* \mid C\cup Z) &>  \Pr(b\gets y \mid C\cup Z)\\
\iff \frac{1}{1 + e^{\log 2}} &> \frac{e^{\log 2}}{1 + e^{\log 2}} \cdot \frac{e^{\log(t-\varepsilon)}}{e^{\log(t-\varepsilon)} + \sum_{z \in Z} e^{\log z}}\\
\iff \frac{1}{3} &> \frac{2}{3} \cdot \frac{t-\varepsilon}{t-\varepsilon + \sum_{z \in Z} z}\\
\iff \sum_{z \in Z} z &> t-\varepsilon.
\end{align*}
Since the $z$ are all integers, we must then have $\sum_{z \in Z} z = t$. If we could efficiently promote $x^*$, we could efficiently find such a $Z$.
\end{proof}

This nested logit construction relies on the two individuals having different nesting structures: notice that $x^*$ and $y$ are swapped in the two trees.
We will see in \cref{sec:restricted} that this is a necessary condition for the hardness of \promo under the nested logit model.

Finally, we have the following hardness result for EBA. 

\begin{theorem}\label{thm:eba_promo_hard}
  In the EBA model, \promo is NP-hard, even with just two individuals
  and two items in $C$.
\end{theorem}

\begin{proof}
By reduction from \textsc{Subset Sum}. Let $S, t$ be an instance of \textsc{Subset Sum}. Let $A=\{a, b\}$, $C = \{x^*, y\}$, $\overline C=S$, and $s = \sum_{z\in S} z$. Make aspects $\chi_z, \psi_z, \gamma_z$ for each $z \in S$ as well as two more aspects $\chi, \psi$. The items have aspects as follows:
  \begin{align*}
    x^{*\prime} &= \{\chi\} \cup \{\chi_z\mid z \in S\}\\
    y' &= \{\psi\} \cup \{\psi_z\mid z \in S\}\\
    z' &= \{\chi_z, \psi_z, \gamma_z\} &&\forall z \in S
  \end{align*}
  The individuals have the following utilities on aspects, where $0 < \varepsilon < 1$:
  \begin{equation*}
  \begin{aligned}[c]
    u_a(\chi) &= 0\\
    u_a(\chi_z) &= z\\
    u_a(\psi) &= s - t/2 - \varepsilon\\
    u_a(\psi_z) &= 0\\
    u_a(\gamma_z) &= s-z
  \end{aligned}
  \;
  \begin{aligned}[c]
    u_b(\chi) &= s-t/2 + \varepsilon\\
    u_b(\chi_z) &= 0\\
    u_b(\psi) &= 0\\
    u_b(\psi_z) &= z\\
    u_b(\gamma_z) &= s-z
  \end{aligned}
  \;
  \begin{aligned}[c]
    &\\
    &\forall z \in S\\
    &\\
    &\forall z \in S\\
    &\forall z \in S
  \end{aligned}
  \end{equation*}
  We want to promote $x^*$. Notice that $x^*$ and $y$ have disjoint aspects. Thus the choice probabilities from $C$ are proportional to the sum of the item's aspects:
  \begin{align*}
  \Pr(a \gets x^* \mid C) &\propto s \phantom{\frac{t}{2}}\\
  \Pr(a \gets y \mid C) &\propto s - \frac{t}{2} - \varepsilon\\
  \Pr(b \gets x^* \mid C) &\propto s - \frac{t}{2} + \varepsilon\\
  \Pr(b \gets y \mid C) &\propto s. \phantom{\frac{t}{2}}
  \end{align*}
  
  To promote $x^*$, we need to make $b$ prefer $x^*$ to $y$. Adding a $z$ item cannibalizes from $a$'s preference for $x^*$ and $b$'s preference for $y$. As in the NL proof, we want to add just enough $z$ items to make $b$ prefer $x^*$ to $y$ without making $a$ prefer $y$ to $x^*$.
   
First, notice that the $\gamma_z$ aspects have no effect on the individuals' relative preference for $x^*$ and $y$.
If we introduce the alternative $z$, then if $a$ picks the aspect $\chi_z$, $y$ will be eliminated. The remaining aspects of $x^*$, namely $x^{*\prime} \setminus \{\chi_z\}$, have combined utility $s -z$, as does $\gamma_z$. Therefore $a$ will be equallly likely to pick $x^*$ and $z$. Symmetric reasoning shows that if $b$ chooses aspect $\psi_z$, then $b$ will end up picking $y$ with probability 1/2. This means that when we include alternatives $Z\subseteq \overline C$, each aspect $\chi_z, \psi_z$ for $z\in Z$ effectively contributes $z/2$ to $a$'s utility for $x^*$ and $b$'s utility for $y$ rather than the full $z$. The optimal solution is therefore a set $Z$ of alternatives whose sum is $t$, since that will cause $a$ to have effective utility $s - t/2$ on $x^*$, which exceeds its utility $s - t/2-\varepsilon$ on $y$. Meanwhile, $b$'s effective utility on $y$ will also be $s-t/2$, which is smaller than its utility $s-t/2+\varepsilon$ on $x^*$. If we include less alternative weight, $b$ will prefer $y$. If we include more, $a$ will prefer $y$. 
  Therefore if we could efficiently find the optimal set of alternatives to promote $x^*$, we could efficiently find a subset of $S$ with sum $t$.
\end{proof}

\subsection{Restricted models that make promotion easier}\label{sec:restricted}
We now show that, in some sense, \promo is a fundamentally easier problem than \agree or \disagree.
Specifically, there are simple restrictions on CDM, NL, and EBA that make \promo tractable but leave \agree and \disagree NP-hard. 
Importantly, these restrictions still allow for choice set effects.
In \cref{sec:stubbornness}, we also prove a strong restriction on the MNL model where \agree and \disagree are tractable,
but we could not find meaningful restrictions for similar results on the other models.

\xhdr{2-item CDM with equal context effects}
The proof of \cref{thm:cdm_promo_hard} shows that \promo is hard with only a single individual and three items in $C$.
However, if $C$ only has two items \emph{and} context effects are the same (i.e., $p_a(z, \cdot)$ is the same for all $z \in \overline C$), 
then \promo is tractable.
The optimal solution is to include all alternatives that increase utility for $x^*$ more than the other item, as doing so makes strict progress on promoting $x^*$.
If individuals have different context effects or if there are more than two items, then there can be conflicts between which items should be included (see \cref{sec:cdm_promo_hard_C_2} for a proof that 2-item CDM with \emph{unequal} context effects makes \promo NP-hard).
Although this restriction makes \promo tractable, it leaves \agree and \disagree NP-hard:
the proofs of \cref{thm:mnl_agree_hard,thm:mnl_disagree_hard} can be interpreted as 2-item and 1-item CDMs
with equal (zero) context effects.

\xhdr{Same-tree NL}
If we require that all individuals share the same NL tree structure, but still allow different utilities, then promotion becomes tractable.
For each $z \in \overline C$, we can determine whether it reduces the relative choice probability of $x^*$ based on its position in the tree:
adding $z$ decreases the relative choice probability of $x^*$ if and only if $z$ is a sibling of any ancestor of $x^*$ (including $x^*$) or if it causes such a sibling to be added to $T_a(C)$.
Thus, the solution to \promo is to include all $z$ not in those positions, which is a polynomial-time check.
This restriction leaves \agree and \disagree NP-hard via \cref{thm:mnl_agree_hard,thm:mnl_disagree_hard} as we can still encode any MNL model in a same-tree NL using the tree in which all items are children of the root.

\xhdr{Disjoint-aspect EBA}
The following condition on aspects makes promoting $x^*$ tractable: 
for all $z \in \overline C$, either $z'\cap x^{*\prime} = \emptyset$ or $z'\cap y' = \emptyset$ for all $y \in C$, $y \ne x^*$.
That is, alternatives either share no aspects with $x^*$ or share no aspects with other items in $C$. This prevents alternatives from cannibalizing from both $x^*$ and its competitors.
To promote $x^*$, we include all alternatives that share aspects with competitors of $x^*$ but not $x^*$ itself, 
which strictly promotes $x^*$.
This condition is slightly weaker than requiring all items to have disjoint aspects, which reduces to MNL.
However, this condition is again not sufficient to make \agree and \disagree tractable, since any MNL model can be encoded in a disjoint-aspect EBA instance. 

\subsection{Strong restriction on MNL that makes \agree and \disagree tractable}\label{sec:stubbornness}

As we saw in the proofs of \cref{thm:mnl_agree_hard,thm:mnl_disagree_hard}, \agree and \disagree are hard in the MNL model even when individuals have identical utilities on alternatives.
This is possible because the individuals have different sums of utilities on $C$; one unit of utility on an alternative has a weaker effect for individuals with higher utility sums on $C$. 
To address the issue of identifiability, we assume each individual's utility sum over $\mathcal U$ is zero in this section. 
This allows us to meaningfully compare the sum of utilities of two different individuals.

\begin{definition}
If an individual $a$ has $\sum_{x\in \mathcal U} u_a(x) = 0$, then the \emph{stubbornness} of $a$ is 
$\sigma_a=\sum_{x \in C}e^{u(x)}.$
\end{definition}
 We call this quantity ``stubbornness'' since it quantifies how reluctant an individual is to change its probabilities on $C$ given a unit of utility on an alternative. 

\begin{proposition}
In an MNL model where all individuals are equally stubborn \emph{and} have identical utilities on alternatives, the solution to \agree is $\overline C$.
\end{proposition}

\begin{proof}
Assume utilities are in standard form, with $\sum_{x\in \mathcal U} u_a(x) = 0$. Let $\sigma=\sum_{x \in C}e^{u(x)}$ be each individual's stubborness and let $Z$ be a set of alternatives. Suppose all individuals have the same utility $u(z)$ for each alternative $z$. The disagreement between two individuals about a single item $x$ in $C$ is then:
\begin{align*}
&\Big| \frac{e^{u_a(x)}}{\sigma + \sum_{z\in Z}e^{u(z)}} - \frac{e^{u_b(x)}}{\sigma + \sum_{z\in Z}e^{u(z)}}\Big|
=\frac{|e^{u_a(x)}-e^{u_b(x)}|}{\sigma + \sum_{z\in Z}e^{u(z)}}.
\end{align*}
Notice that this strictly decreases if $\sum_{z\in Z}e^{u(z)}$ increases, so we minimize $D$ by including all of the alternatives.
\end{proof}

The same reasoning also allows us to trivially solve \disagree in this restricted MNL model.

\begin{corollary}
The solution to \disagree in an equal alternative utilities, equal stubbornness MNL model is $\emptyset$.
\end{corollary}

 While this MNL restriction is too strong to be of practical value, it is interesting from a theoretical perspective as it indicates where the hardness of the problem arises.

\section{Approximation algorithms}

\begin{algorithm}[tb]
   \caption{$\varepsilon$-additive approximation for \agree in the MNL model.}
   \label{alg:approx}
\begin{algorithmic}[1]
\algsetup{
  indent=1.5em,
  linenosize=\scriptsize,
  linenodelimiter=
}
  \STATE {\bfseries Input:} $n$ individuals $A$, $k$ items $C$, $m$ alternatives $\overline C$, utilities $u_a(\cdot) > 0$ for each $a \in A$. For brevity:
  \STATE $e_{ax} \gets e^{u_a(x)}$,\; $s_a \gets \sum_{z\in \overline C} e_{az}$,\; $\delta \gets \varepsilon / (2km\binom{n}{2})$ \label{line:delta}
  \STATE $L_0 \gets$ empty $n$-dimensional array whose $a$th dimension has size $1+\lfloor \log_{1+\delta} s_a \rfloor$ (each cell can store a set $Z\subseteq \overline C$ and its $n$ exp-utility sums for each individual)
  \STATE Initialize $L_0[0, \dots, 0] \gets (\emptyset, 0, \dots, 0)$ ($n$ zeros)
  \FOR{$i=1$ {\bfseries to} $m$}
    \STATE $z \gets \overline C[i-1]$,\;\; $L_{i} \gets L_{i-1}$
    \FOR{\textbf{each} cell of $L_{i-1}$ containing $(Z, t_1, \dots, t_n)$}
      \STATE $h\gets$ $n$-tuple w/ entries $\lfloor\log_{1+\delta} (t_j + e_{a_jz})\rfloor$, $\forall j$

      \IF{$L_{i}[h]$ is empty}
        \STATE $L_i[h]\gets (Z\cup\{z\}, t_1 + e_{a_1z}, \dots, t_n + e_{a_nz})$
      \ENDIF
    \ENDFOR
  \ENDFOR

  \STATE $Z_m \gets$ collection of all sets $Z$ in cells of $L_m$
  \RETURN $\arg\min_{Z \in Z_m} D(Z)$ (see \cref{eq:D}) \label{line:return}
\end{algorithmic}
\end{algorithm}

Thus far, we have seen that several interesting group decision-making problems are NP-hard across standard discrete choice models.
Here, we provide a positive result: we can compute arbitrarily good approximate solutions to many instances of these problems in polynomial time.
We focus our analysis on \cref{alg:approx}, which is an $\varepsilon$-additive approximation algorithm to \agree under MNL,
with runtime polynomial in $k$, $m$, and $\frac 1 \varepsilon$, but exponential in $n$ (recall that $k = |C|$, $m=|\overline C|$, and $n=|A|$). In contrast, brute force (testing every set of alternatives) is exponential in $m$ and polynomial in $k$ and $n$.
\agree is NP-hard even with $n = 2$ (\cref{thm:mnl_agree_hard}), so our algorithm provides a substantial efficiency improvement.
We discuss how to extend this algorithm to other objectives and other choice models later in the section. Finally, we present a faster but less flexible mixed-integer programming approach for MNL \agree and \disagree that performs very well in practice.

\Cref{alg:approx} is based on an FPTAS for \subsetsum~\citep[Sec.~35.5]{clrs}, and the first parts of our analysis follow some of the same steps.
The core idea of our algorithm is that a set of items can be characterized by its exp-utility sums for each individual and that there are only polynomially many combinations of exp-utility sums that differ by more than a multiplicative factor $1+\delta$.
We can therefore compute all sets of alternatives with meaningfully different impacts and pick the best one.
For the purpose of the algorithm, we assume all utilities are positive (otherwise we may access a negative index);
utilities can always be shifted by a constant to satisfy this requirement. 

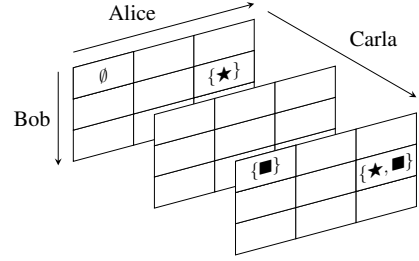
\begin{figure}
\centering
\scalebox{0.83}{
  \begin{tikzpicture}[x=(15:.5cm), y=(90:.5cm), z=(330:.5cm), >=stealth]
\foreach \z in {0, 3, 6} 
\foreach \x in {0,2, 4}
  \foreach \y in {0,...,2}
    \filldraw [fill=white] (\x, \y, \z) -- (\x+2, \y, \z) -- (\x+2, \y+1, \z) -- (\x, \y+1, \z) -- cycle (\x+1, \y+.5, \z) node [yslant=tan(15)] {};
\filldraw [fill=white] (0, 2, 0) -- (2, 2, 0) -- (2, 3, 0) -- (0, 3, 0) -- cycle (1, 2.5, 0) node [yslant=tan(15)] {\scalebox{0.8}{$\emptyset$}};
\filldraw [fill=white] (4, 1, 0) -- (6, 1, 0) -- (6, 2, 0) -- (4, 2, 0) -- cycle (5, 1.5, 0) node [yslant=tan(15)] {\scalebox{0.8}{$\{\bigstar\}$}};

\filldraw [fill=white] (4, 1, 6) -- (6, 1, 6) -- (6, 2, 6) -- (4, 2, 6) -- cycle (5, 1.5, 6) node [yslant=tan(15)] {\scalebox{0.8}{$\{\bigstar, \blacksquare\}$}};

\filldraw [fill=white] (0, 2, 6) -- (0, 3, 6) -- (2, 3, 6) -- (2, 2, 6) -- cycle (1, 2.5, 6) node [yslant=tan(15)] {\scalebox{0.8}{$\{\blacksquare\}$}};

\draw [->] (0, 3.5, 0)  -- (6, 3.5, 0)   node [midway, above left] {Alice};
\draw [->] (-.5, 3, 0)  -- (-.5, 0, 0)   node [midway, left] {Bob};
\draw [->] (6, 3.5, 0.5) -- (6, 3.5, 6) node [midway, above right] {Carla};
\end{tikzpicture}
}
\vspace{-3.5mm}
\caption{Example of the structure $L_i$ used in \Cref{alg:approx} for $n=3$ individuals and $\overline C = \{\bigstar, \blacksquare\}$. Here, Alice has high utility for $\bigstar$ and low utility for $\blacksquare$, Bob has medium utility for $\bigstar$ and low utility for $\blacksquare$, and Carla has low utility for $\bigstar$ and high utility for $\blacksquare$. The exp-utility sums stored in cells are omitted.}
\label{fig:L_diag}
\end{figure}

We now provide an intuitive description of \Cref{alg:approx}. 
The array $L_i$ has one dimension for each individual in $A$ (we use a hash table in practice since $L_i$ is typically sparse). The cells along a particular dimension discretize the exp-utility sums that the individual corresponding to that dimension could have for a particular set of alternatives (\Cref{fig:L_diag}). In particular, if individual $j$ has total exp-utility $t_j = \sum_{y \in Z} e^{u_j(y)}$ for a set $Z$, then we  store $Z$ at index $\lfloor\log_{1+\delta} t_j \rfloor$ along dimension $j$. 

As the algorithm progresses, we place possible sets of alternatives $Z$ in the cells of $L_i$ according to their exp-utility sums $t_1, \dots, t_n$ for each individual (we store $t_1, \dots, t_n$ in the cell along with $Z$). We add one element at a time from $\overline C$ to the sets already in $L_i$ ($L_0$ starts with only the empty set). If two sets have very similar exp-utility sums, they may map to the same cell, in which case only one of them is stored. If the discretization of the array is coarse enough (that is, with large enough $\delta$), many sets of alternatives will map to the same cells, reducing the number of sets we consider and saving computational work. On the other hand, if the discretization is fine enough ($\delta$ is sufficiently small), then the best set we are left with at the end of the algorithm cannot induce a disagreement value too different from the optimal set. We now formalize this reasoning, starting with the following technical lemma that shows items mapping to same cell have similar exp-utility sums.

\begin{lemma}\label{lemma:exp_utility_bounds}
Let $\overline C_i$ be the first $i$ elements processed by the outer for loop of \Cref{alg:approx}. At the end of the algorithm, for all $Z \subseteq \overline C_i$ with exp-utility sums $t_a$, there exists some $Z' \in L_i$ with exp-utility sums $t_a'$ such that 
$\frac{t_a}{(1+\delta)^i} < t_a' < t_a(1+\delta)^i$, for all $a \in A$ (with $\delta$ as defined in \Cref{alg:approx}, Line~\ref{line:delta}).
\end{lemma}

\begin{proof}
If a set $Z$ has total exp-utility $t_a$ to individual $a$, then it is placed in $L$ at position $\lfloor \log_{1+\delta} t_a\rfloor$ in dimension $a$. So, if two sets $Z$, $Z'$ with exp-utility totals $t_a, t_a'$ for individual $a$ are mapped to the same cell of $L$, then for all $a\in A$, $\lfloor \log_{1+\delta} t_a\rfloor = \lfloor \log_{1+\delta} t_a'\rfloor$. We can therefore bound $t_a'$: 
\begin{align*}
\log_{1+\delta} t_a - 1 < \log_{1+\delta} t_a' < \log_{1+\delta} t_a +1.
\end{align*}
Exponentiating both sides with base $1+\delta$ and simplifying yields
\begin{equation}
\frac{t_a}{1+\delta} <  t_a' < t_a(1+\delta). \label{eq:mult_bounds}
\end{equation}

With this fact in hand, we proceed by induction on $i$. When $i=0$, $\overline C_i$ is empty and the lemma holds. Now suppose that $i > 0$ and that the lemma holds for $i-1$. Every set in $\overline C_i$ was made by adding (a) zero elements or (b) one element to a set in $\overline C_{i-1}$. We consider these two cases separately. 

(a) For any set $Z\subseteq \overline C_{i}$ that is also contained in $\overline C_{i-1}$, we know by the inductive hypothesis that some element in $L_{i-1}$ satisfied the inequality. Since we never overwrite cells, the lemma also holds for $Z$ after iteration $i$. 

(b) Now consider sets $Z'\subseteq \overline C_{i}$ that were formed by adding the new element $z$ to a set $Z \subseteq \overline C_{i-1}$. In the inner for loop, we at some point encountered the cell containing the set $Y \in L_{i-1}$ satisfying the lemma for set $Z$ by the inductive hypothesis. Let $y_a$ be the exp-utility totals for $Y$ and $t_a$ for $Z$. Notice that the exp-utility totals of $Z'$ are exactly $t_a + e_{az}$. Starting with the inductive hypothesis, we see that the exp-utility totals of $Y\cup \{z\}$ satisfy
\begin{align*}
\frac{t_a + e_{az}}{(1+\delta)^{i-1}} < y_a+ e_{az} < (t_a+ e_{az})(1+\delta)^{i-1}.
\end{align*}    

When we go to place $Y\cup \{z\}$ in a cell, it might be unoccupied, in which case we place it in $L_i$ and the lemma holds for $Z'$. If it is occupied by some other set, then by applying \cref{eq:mult_bounds} we find that the lemma also holds for $Z'$.
\end{proof}

With this lemma in hand, we can prove our main constructive result.

\begin{theorem}\label{thm:approx_alg}
\Cref{alg:approx} is an $\varepsilon$-additive approximation for \agree in the MNL model.
\end{theorem}

\begin{proof}
Let $\beta = \varepsilon / (k\binom{n}{2})$ for brevity.
Following our choice of $\delta$ and using \cref{lemma:exp_utility_bounds},
at the end of the algorithm, the optimal set $Z^*\subseteq \overline C$ (with exp-utility sums $t_a^*$) has some representative $Z'$ in $L_m$ such that
\[
 \frac{t_a^*}{(1+\beta / (2m))^m} < t_a' < t_a^*\left(1+\beta / (2m)\right)^m,\;  \forall a \in A.
\]
Since $e^x \ge (1+ x/m)^m$, we have $t_a^* / e^{\frac{\beta}{2}} < t_a' < t_a^* e^\frac{\beta}{2}$,
and since $e^x \le 1+x+x^2$ when $x< 1$,
\[
 \frac{t_a^*}{1+\beta / 2 + \beta^2 / 4} < t_a' < t_a^* (1+ \beta / 2 + \beta^2 / 4).
\]
Finally, $\frac{t_a^*}{1+\beta} < t_a' < t_a^* (1+\beta)$ because $0 < \beta < 1$.

Now we show that $D(Z^*)$ and $D(Z')$ differ by at most $\varepsilon$. To do so, we first bound the difference between $\Pr(a\gets x \mid C\cup Z^*)$ and $\Pr(a\gets x \mid C\cup Z')$ by $\beta$. Let $c_a = \sum_{x\in C}e_{ax}$ be the total exp-utility of $a$ on $C$. By the above reasoning,
\[
 \frac{e_{ax}}{c_a + t_a^* (1+\beta)} < \frac{e_{ax}}{c_a + t_a'} <\frac{e_{ax}}{c_a + \frac{t_a^*}{1+\beta}},
\]
where the middle term is equal to $\Pr(a\gets x \mid C\cup Z')$.
From the lower bound, the difference between $\Pr(a\gets x \mid C\cup Z^*)$ and $\Pr(a\gets x \mid C\cup Z')$ could be as large as
\begin{align*}
  &\frac{e_{ax}}{c_a + t_a^*} - \frac{e_{ax}}{c_a + t_a^*(1+\beta)} \\
  &= \frac{e_{ax}t_a^*\beta}{(c_a + t_a^*)(c_a + t_a^*(1+\beta))}
  < \frac{e_{ax}t_a^*\beta}{2c_at_a^*} \le \frac{\beta}{2}. 
\end{align*}
From the upper bound, the difference between $\Pr(a\gets x \mid C\cup Z^*)$ and $\Pr(a\gets x \mid C\cup Z')$ could be as large as
\begin{align*}
  &\frac{e_{ax}}{c_a + \frac{t_a^*}{1+\beta}} - \frac{e_{ax}}{c_a + t_a^*} %
  = \frac{e_{ax}t_a(1-\frac{1}{1+\beta})}{(c_a + \frac{t_a^*}{1+\beta})(c_a + t_a^*)} \\
  &= \frac{e_{ax}t_a^*\beta}{(c_a(1+\beta) + t_a^*)(c_a + t_a^*)} 
  < \frac{e_{ax}t_a^*\beta}{2c_a t_a^*}  \le \frac{\beta}{2}. 
\end{align*}
Thus, $\Pr(a\gets x \mid C\cup Z^*)$ and $\Pr(a\gets x \mid C\cup Z')$ differ by at most $\frac{\beta}{2}$. 
Using the same argument for an individual $b$, the disagreement between $a$ and $b$ about $x$ can only increase by $\beta$ with the set $Z$ compared to the optimal set $Z^*$. Since there are $\binom{n}{2}$ pairs of individuals and $k$ items in $C$, the total error of the algorithm is bounded by $k\binom{n}{2}\beta = \varepsilon$.
\end{proof}

We now show that the runtime of \cref{alg:approx} is $O((m+ kn^2)(1+\lfloor \log_{1+\delta} s \rfloor)^n)$, where $s = \max_a s_a$ is the maximum exp-utility sum for any individual. Moreover, for any fixed $n$, this runtime is bounded by a polynomial in $k, m$, and $\frac{1}{\varepsilon}$.  

To see this, first note that the size of $L_i$ is bounded above by $(1+\lfloor \log_{1+\delta} s \rfloor)^n$. For each $z \in \overline C$, we perform constant-time operations\footnote{The algorithm requires computing $\log_{1+\delta}$, which can be done efficiently using a precomputed change-of-base constant and taking logarithms to a convenient base. Our analysis treats these logarithms as constant-time operations, since we care about how the runtime grows as a function of $n, m, k, $ and $\frac{1}{\varepsilon}$.} on each entry of $L_i$, for a total of $O(m (1+\lfloor \log_{1+\delta} s \rfloor)^n)$ time. Then we compute $D(Z)$ for each cell of $L_m$, which takes $O(kn^2)$ time per cell. The total runtime is therefore $O((m+kn^2)(1+\lfloor \log_{1+\delta} s \rfloor)^n)$, as claimed. Finally, $(1+\lfloor \log_{1+\delta} s \rfloor)^n$ is bounded by a polynomial in $m, k$, and $\frac{1}{\varepsilon}$ for any fixed $n$:
\begin{align*}
  (1+\lfloor \log_{1+\delta} s \rfloor)^n &\le \Big(1+ \frac{\ln s}{\ln 1+\delta} \Big)^n\\
  &\le \Big(1+ (1+\delta)\frac{\ln s}{\delta} \Big)^n\tag{since $\ln (1+x) \ge\frac{x}{1+x}$ for $x > -1$}\\
  &= \Big(1+ \frac{\ln s}{\delta} + \ln s\Big)^n\\
  &= \Big(1+ \frac{2km\binom{n}{2}\ln s}{\varepsilon} + \ln s\Big)^n.
\end{align*}

\agree is NP-hard even when individuals have equal utilities on alternatives.
In this case, we only need to compute exp-utility sums for a single individual,
which brings the runtime down to $O((m+ kn^2) \log_{1+\delta} s)$.

\xhdr{Extensions to other objectives and models}
\Cref{alg:approx} can be easily extended to any objective function that is efficiently computable from utilities.
For instance, \cref{alg:approx} can be adapted for \disagree by replacing the $\arg\min$ with an $\arg\max$
on Line~\ref{line:return}.

\Cref{alg:approx} can also be adapted for CDM and NL.
The analysis is similar and details are in \cref{sec:alg_details},
although the running times and guarantees are different.
With CDM, the exponent in the runtime increases to $nk$ for \agree and \disagree,
and the $\varepsilon$-additive approximation is guaranteed only if items in $\overline C$ exert zero pulls on each other.
However, even for the general CDM, our experiments will show that the adapted algorithm remains a useful heuristic.
When we adapt \Cref{alg:approx} for NL, we retain the full approximation guarantee but the exponent in the runtime increases and has a dependence on the tree size.

\promo is not interesting under MNL and also has a discrete rather than continuous objective,
i.e., the number of people with favorite item $x^*$ in $C$.
For models with context effects, we can define a meaningful
notion of approximation.
\begin{definition}
An item $y \in C \cup Z$ is an \emph{$\varepsilon$-favorite} item of individual $a$ if $\Pr(a \gets y \mid C \cup Z) +\varepsilon \ge \Pr(a\gets x \mid C \cup Z)$ for all $x \in C$. A solution \emph{$\varepsilon$-approximates} \promo if the number of people for whom $x^*$ is an $\varepsilon$-favorite item is at least the value of the optimal \promo solution.
\end{definition}  

Using this, we can adapt \Cref{alg:approx} for \promo under CDM and NL.
Again, for CDM, the approximation has guarantees in certain parameter regimes and the NL has full approximation guarantees.
Since we do not have compute $D(Z)$, the runtimes loses the $kn^2$ term compared to the \agree and \disagree versions (\Cref{sec:approx_promo}).

Finally, EBA has considerably different structure than the other models.
We leave algorithms for EBA to future work.

\subsection{Fast exact methods for MNL} 
We provide another approach for solving \agree
and \disagree in the MNL model, based on transforming the objective functions
into mixed-integer bilinear programs (MIBLPs).
MIBLPs can be solved for moderate problem sizes with high-performance
branch-and-bound solvers (we use Gurobi's implementation). In practice, this
approach is faster than \Cref{alg:approx} (for finding optimal solutions---\Cref{alg:approx} will always be faster with sufficiently large $\varepsilon$) and can optimize over larger sets
$\overline C$.  However, this approach does not easily extend to CDM, NL, or
\promo and does not have a polynomial-time runtime guarantee.

\subsubsection{MIBLP formulation for \agree}
Let $x_i$ be a decision variable indicating whether we add in the $i$th item in $\overline C$. Let $e_{ya} = e^{u_a(y)}$ and $e_{Ca} = \sum_{y\in C} e_{ya}$.
 We can write \agree as the following 0-1 optimization problem.
\begin{mini*}{x}{\sum_{a, b \in A} \sum_{y \in C} \left|\frac{e_{ya}}{e_{Ca} + \sum_{i \in \overline C} x_i e_{ia}} - \frac{e_{yb}}{e_{Cb} + \sum_{i \in \overline C} x_i e_{ib}}\right|}{}{}
\addConstraint{x_i \in \{0, 1\}}
\end{mini*}
We can rewrite this with no absolute values by introducing new variables $\delta_{yab}$ that represent the absolute disagreement about item $y$ between individuals $a$ and $b$. We then use the standard trick for minimizing an absolute value in linear programs: 

\begin{mini*}[2]{x}{\sum_{a, b \in A} \sum_{y \in C} \delta_{yab}}{}{}
\addConstraint{\frac{e_{ya}}{e_{Ca} + \sum_{i \in \overline C} x_i e_{ia}} - \frac{e_{yb}}{e_{Cb} + \sum_{i \in \overline C} x_i e_{ib}}}{\le \delta_{yab}\quad}{\breakObjective{\forall y \in C, \{a, b\} \subset A}}
\addConstraint{\frac{e_{yb}}{e_{Cb} + \sum_{i \in \overline C} x_i e_{ib}} - \frac{e_{ya}}{e_{Ca} + \sum_{i \in \overline C} x_i e_{ia}}}{\le \delta_{yab} \quad}{\breakObjective{\forall y \in C, \{a, b\} \subset A}}
\addConstraint{x_i}{\in \{0, 1\}\quad\forall i \in \overline C}{}
\addConstraint{\delta_{yab}}{\in \mathbb{R}\qquad\forall y \in C, \{a, b\} \subset A}{}
\end{mini*}
To get rid of the fractions, we introduce the new variables $z_{a} = \frac{1}{e_{Ca} + \sum_{i} x_i e_{ia}}$ for each individual $a$ and add corresponding constraints enforcing the definition of $z_a$:
\begin{mini*}[1]{x}{\sum_{a, b \in A} \sum_{y \in C} \delta_{yab}}{}{}
\addConstraint{z_a e_{ya} - z_b e_{yb}}{\le \delta_{yab}\quad}{\forall y \in C, \{a, b\} \subset A}
\addConstraint{z_b e_{yb} - z_a e_{ya}}{\le \delta_{yab} \quad}{\forall y \in C, \{a, b\} \subset A}
\addConstraint{z_a e_{Ca} + z_a \sum_{i \in \overline C} x_i e_{ia}}{=1}{\forall a \in A}
\addConstraint{x_i}{\in \{0, 1\}\quad}{\forall i \in \overline C}
\addConstraint{\delta_{yab}}{\in \mathbb{R}\quad}{\forall y \in C, \{a, b\} \subset A}
\addConstraint{z_a}{\in \mathbb{R}}{\forall a \in A}
\end{mini*}

\subsubsection{MIBLP formulation for \disagree}
A similar technique works for \disagree, but maximizing an absolute value is slightly trickier than minimizing. In addition to the variables $\delta_{yab}$ that we used before, we also add new binary variables $g_{yab}$ indicating whether each difference in choice probabilities is positive or negative. With these new variables (and following the same steps as above), \disagree can be written as the following MIBLP:
\begin{maxi*}[1]{x}{\sum_{a, b \in A} \sum_{y \in C} \delta_{yab}}{}{}
\addConstraint{z_a e_{ya} - z_b e_{yb}}{\le \delta_{yab}\quad}{\forall y \in C, \{a, b\} \subset A}
\addConstraint{z_b e_{yb} - z_a e_{ya}}{\le \delta_{yab} \quad}{\forall y \in C, \{a, b\} \subset A}
\addConstraint{2 g_{yab} + z_a e_{ya} - z_b e_{yb}}{\ge \delta_{yab}\quad}{\forall y \in C, \{a, b\} \subset A}
\addConstraint{2 (1-g_{yab}) + z_b e_{yb} - z_a e_{ya}}{\ge \delta_{yab} \quad}{\forall y \in C, \{a, b\} \subset A}
\addConstraint{z_a e_{Ca} + z_a \sum_{i \in \overline C} x_i e_{ia}}{=1}{\forall a \in A}
\addConstraint{x_i}{\in \{0, 1\}\quad}{\forall i \in \overline C}
\addConstraint{g_{yab}}{\in \{0, 1\}\quad}{\forall y \in C, \{a, b\} \subset A}
\addConstraint{\delta_{yab}}{\in \mathbb{R}\quad}{\forall y \in C, \{a, b\} \subset A}
\addConstraint{z_a}{\in \mathbb{R}}{\forall a \in A}
\end{maxi*}

\section{Numerical experiments}\label{sec:results}

We apply our methods to three datasets (\cref{tab:dataset_summary}). 
The \sfwork dataset~\cite{koppelman2006self} comes from a survey of San Francisco residents on available (choice set) and 
selected (choice) transportation options to get to work. We split the respondents into two segments ($\lvert A \rvert = 2$) according to whether or not they live in the ``core residential district of San Fransisco or Berkeley.''
The \allstate dataset~\cite{allstate_data} consists of insurance policies (items) characterized by anonymous categorical features A--G with 2 to 4 values each. Each customer views a set of policies (the choice set) before purchasing one. We reduce the number of items to 24 by considering only features A, B, and C. To model different types of individuals, we split the data into homeowners and non-homeowners (again, $\lvert A \rvert = 2$). 
The \yoochoose dataset~\cite{ben2015recsys} contains online shopping data
of clicks and purchases of categorized items in user browsing sessions.
Choice sets are unique categories browsed in a session and the choice is the category of the purchased product
(categories appearing fewer than 20 times were omitted).
We split the choice data into two sub-populations by thresholding on the purchase timestamps.

For inferring maximum-likelihood models from data, we use PyTorch's Adam optimizer~\cite{kingma2015adam,paszke2019pytorch} with learning rate $0.05$, weight decay $0.00025$, batch size 128, and the \texttt{amsgrad} flag~\cite{reddi2018convergence}. 
We use the low-rank (rank-2) CDM~\cite{seshadri2019discovering} that expresses pulls as the inner product of item embeddings. Our code and data are available at \url{https://github.com/tomlinsonk/choice-set-opt}. 

\begin{table}[t]
\centering
\caption{Dataset statistics: item, observation, and unique choice set counts; and percent of observations in sub-population splits.}
\label{tab:dataset_summary}
\begin{tabular}{l r r r r}
\toprule
Dataset & \# items & \# obs. & \# sets & split \% \\
\midrule
\sfwork & 6 & 5029 & 12 & 16/84\\
\allstate & 24 & 97009 & 2697 & 45/55\\
\yoochoose & 41 & 90493 & 1567 & 47/53 \\
\bottomrule
\end{tabular}
\end{table}

\begin{table}[t]
\centering
\caption{Sum of error over all 2-item choice sets $C$ compared to optimal (brute force) on \sfwork.
\Cref{alg:approx} is optimal.}\label{tab:sfwork_summary}
\begin{tabular}{l l c c c c}
 \toprule
 Model & Problem & Greedy & \Cref{alg:approx} \\
 \midrule
 MNL & \agree & $0.03$ & $ 0.00$ \\
        & \disagree & $0.00$ & $0.00$ \\
  rank-2 CDM & \agree & $0.14$ &$ 0.00$ \\
 & \disagree  &  $0.13$ & $0.00$ \\
  NL & \agree &  $0.00$ & $0.00$ \\
 & \disagree  &  $0.00$  &$0.00$ \\
 \bottomrule
\end{tabular}
\end{table}

For \sfwork under the MNL, CDM, and NL models, we considered all 2-item choice sets $C$ (using all other items for $\overline C$) 
for \agree and \disagree (for the NL model, we used the best-performing tree from \citet{koppelman2006self}).
We compare \Cref{alg:approx} ($\varepsilon = 0.01$) to a greedy approach (henceforth, ``Greedy'') that builds $Z$ by repeatedly selecting the item from $\overline C$ that, when added to $Z$, most improves the objective, if such an item exists.
This dataset was small enough to compare against the optimal, brute-force solution (\cref{tab:sfwork_summary}). 
In all cases, \cref{alg:approx} finds the optimal solution, while Greedy is often suboptimal. 
However, for this value of $\varepsilon$, we find that \cref{alg:approx} searches the entire
space and actually computes the brute force solution (we get the number of sets analyzed by \cref{alg:approx} from $|L_m|$ for a given $\varepsilon$ and compare it to $2^{\lvert \overline C \rvert}$).
Even though we have an asymptotic polynomial runtime guarantee, for small enough datasets, we might not see computational savings.
Running with larger $\varepsilon$ yielded similar results, even for $\varepsilon > 2$, when our bounds are vacuous.

The results still highlight two important points.
First, even on small datasets, Greedy can be sub-optimal. 
For example, for \agree under CDM with $C=\{\text{drive alone}, \text{transit}\}$, \Cref{alg:approx} found the optimal $Z = \{\text{bike, walk}\}$, inferring that both sub-populations agree on both driving less and taking transit less. 
However, Greedy just introduced a carpool option, which has a lower effect on discouraging driving alone or taking transit, resulting in lower agreement between city and suburban residents.

\begin{figure}[t]
    \centering
    \includegraphics[width=\columnwidth]{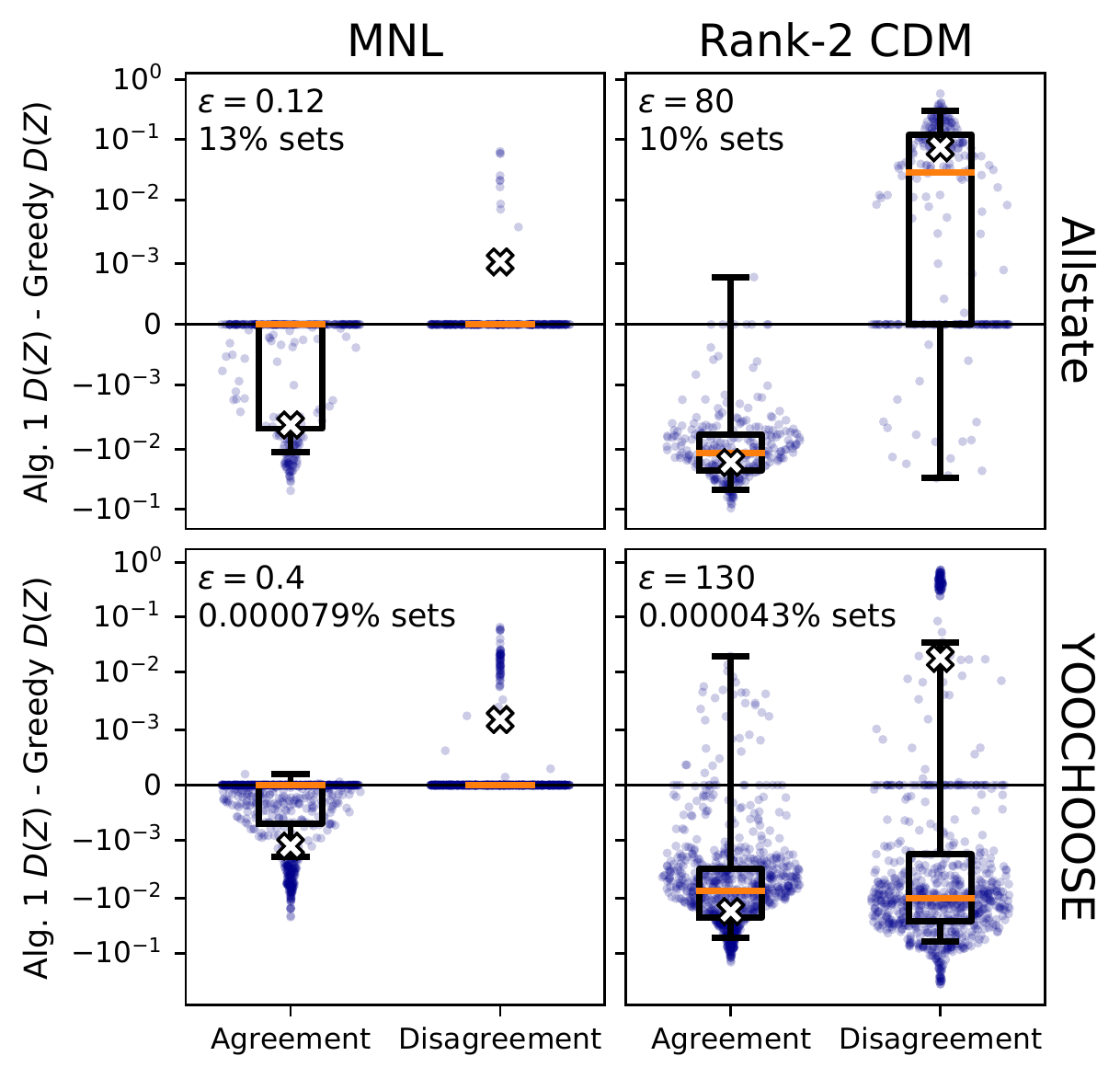}
    \vspace{-9mm}
    \caption{\Cref{alg:approx} vs.\ Greedy performance box plots
    when applied to all 2-item choice sets in \allstate and \yoochoose under MNL and CDM
    (subplots also show $\varepsilon$ and the percent of subsets of $\overline C$ computed by \cref{alg:approx}, written X\% sets).
    Each point is the difference in $D(Z)$ when \Cref{alg:approx} and Greedy are run on a particular choice set. 
    Horizontal spread shows approximate density and the Xs mark means. 
     A negative (resp.~positive) $y$-value for \agree (resp.~\disagree) indicates that \Cref{alg:approx} outperformed Greedy. 
     \Cref{alg:approx} performs better in all cases except for \disagree under CDM on \yoochoose.
     Even in this exception, though, our approach finds a few very good solutions and \cref{alg:approx}
     has better mean performance.}
    \label{fig:all_pairs}
\end{figure}

Second, our theoretical bounds can be more pessimistic than what happens in practice.
Thus, we can consider larger values of $\varepsilon$ to reduce the search space;
 \cref{alg:approx} remains a principled heuristic, and we can measure 
how much of the search space \cref{alg:approx} considers.
This is the approach we take for the \allstate and \yoochoose data,
where we find that \cref{alg:approx} far outperforms its theoretical worst-case bound.
We again considered all 2-item choice sets $C$ and tested our method under MNL and CDM,\footnote{In this case, we did
not have available tree structures for NL, which are difficult to derive from data~\cite{benson2016relevance}.} 
setting $\varepsilon$ so that the experiment took about 30 minutes to run for \allstate and 2 hours for \yoochoose (of that time, Greedy takes 5 seconds to run; the rest is taken up by \cref{alg:approx}).
\Cref{alg:approx} consistently outperforms Greedy (\cref{fig:all_pairs}), even with $\varepsilon > 2$ for CDM. 
Moreover, \cref{alg:approx} only computes a small fraction of possible sets of alternatives, especially for \yoochoose.
\Cref{alg:approx} does not perform as well with the rank-2 CDM as it does with MNL, which is to be expected as we only have approximation guarantees for CDM under particular parameter regimes (in which these data do not lie). The worse performance on CDM is due to the context effects that items from $\overline C$ exert on each other. 
Greedy does fairly well for \disagree under CDM with \yoochoose, but even in this case, \cref{alg:approx} performs significantly better in enough instances for the mean (but not median) performance to be better than Greedy. 
We repeated the experiment with 500 choice sets of size up to 5 sampled from data with similar results 
(\cref{sec:sampled_choice_sets}). We also ran the MIBLP approach for MNL,
which performed as well as \Cref{alg:approx} and was about $12$x faster on \yoochoose and $240$x faster on \allstate with the $\varepsilon$ values we used for \Cref{alg:approx}
(\Cref{sec:MIBLP_results}). 


\begin{figure}[tb]
    \centering
    \includegraphics[width=0.49\columnwidth]{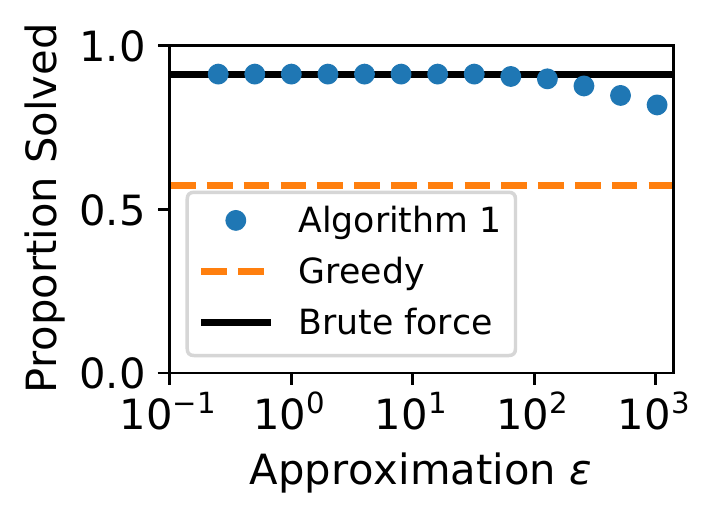}
    \hfill
    \includegraphics[width=0.49\columnwidth]{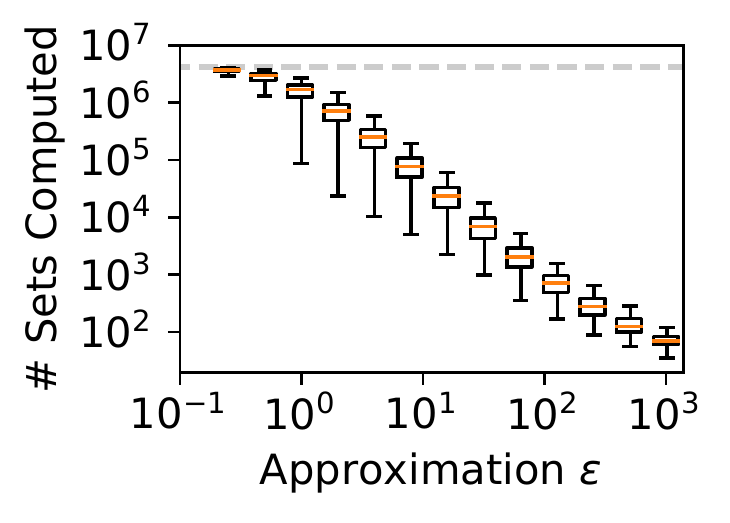}
    \vspace{-5mm}
    \caption{\promo results on \allstate 2-item choice sets. 
    (Left) Success rate comparison; \cref{alg:approx} has near-optimal performance (about $9\%$ of instances have no \promo solution).
    (Right) Number of subsets of $\overline C$ computed by \cref{alg:approx} (dashed gray line at $2^{22} = 2^m$ for brute force computation).}
    \label{fig:promo_allstate}
\end{figure}

\xhdr{\promo} We applied the CDM \promo version of \cref{alg:approx} to \allstate, since this dataset is small enough to compute brute-force solutions.
For each 2-item choice set $C$, we attempted to promote the less-popular item of the pair using brute-force, Greedy, and \cref{alg:approx}. 
\Cref{alg:approx} performed optimally up to $\varepsilon=32$, above which it failed in only 2--26 of 252 feasible instances (\cref{fig:promo_allstate}, left).
(Here, successful promotion means that the item becomes the true favorite among $C$.)
On the other hand, Greedy failed in $37\%$ of the feasible instances.
As in the previous experiment, our algorithm's performance in practice far exceeds the worst-case bounds. 
The number of sets tested by \cref{alg:approx} falls dramatically as $\varepsilon$ increases (\cref{fig:promo_allstate}, right).
With more items (or a smaller range of utilities), the value of $\varepsilon$ required to achieve the same speedup over brute force would be smaller (as with \yoochoose). 
In tandem, these results show that we get near-optimal \promo performance with far fewer computations than brute force. 

\section{Discussion}
Our decisions are influenced by the alternatives that are available, the choice set.
In collective decision-making, altering the choice set can encourage agreement 
or create new conflict.
We formulated this as an algorithmic question: how can we optimize the choice set for some objective?

We showed that choice set optimization is NP-hard for natural objectives under standard choice models;
however, we also found that model restrictions makes promoting a choice easier than encouraging a group to agree or disagree.
We developed approximation algorithms for these hard problems that are effective in practice,
although there remains a gap between theoretical approximation bounds and performance on real-world data. 
Future work could address choice set optimization in interactive group decisions, where group members can communicate their preferences to each other or must collaborate to reach a unified decision. Lastly, \cref{sec:ethics} discusses the ethical considerations of this work.

\section*{Acknowledgments}

This research was supported by
ARO MURI,
ARO Award W911NF19-1-0057,
NSF Award DMS-1830274, and
JP Morgan Chase \& Co.
We thank Johan Ugander for helpful conversations.

\bibliography{references}
\bibliographystyle{icml2020}

\clearpage
\appendix

\section{Hardness proofs}\label{sec:proofs}

\subsection{CDM \promo is hard with $|A|=2, |C|=2$}
\label{sec:cdm_promo_hard_C_2}
In the main text, we show CDM \promo is NP-hard when $|A|=1, |C|=3$ (\cref{thm:cdm_promo_hard}). Here, we provide an additonal proof for the case when $|A|=2, |C|=2$. These are the smallest hard instances of the problem ($|A|=1, |C|=2$ is easy to solve: introduce alternatives that increase utility for $x^*$ for than its competitor).

\begin{theorem}
  In the CDM model, \promo is NP-hard,
  even with just two individuals and two items in $C$.
\end{theorem}

\begin{proof}
By reduction from \textsc{Subset Sum}. Let $S, t$ be an instance of \textsc{Subset Sum}. Let $A=\{a, b\}$, $C = \{x, y\}$, $\overline C=S$. Using tuples interpreted entrywise, construct a CDM with the following parameters.
\begin{align*}
  u_{a}(\langle x^*, y \rangle) &=  \langle t+\varepsilon,0  \rangle \\
  u_{b}(\langle x^*, y \rangle) &=  \langle \varepsilon, t  \rangle \\
  u_a(z) &= u_b(z) = -\infty &&\forall z \in \overline C\\
  p_{a}(z, \langle x^*, y\rangle ) &=  \langle 0, z \rangle  &&\forall z \in \overline C\\
  p_{b}(z, \langle x^*, y\rangle ) &=  \langle z, 0 \rangle  &&\forall z \in \overline C
\end{align*}
To promote $x^*$, we need to add more than $t-\varepsilon$ to $b$'s utility for $x^*$, but add less than $t+\varepsilon$ to $a$'s utility for $x^*$. Since all pulls are integral, the only solution is a set $Z$ whose sum of pulls is $t$. If we could efficiently find such a set, then we could efficiently find the \subsetsum solution. 
\end{proof}

\section{Approximation algorithm extensions}\label{sec:alg_details}

\subsection{Adapting \Cref{alg:approx} for CDM  with guarantees for special cases}\label{sec:app_approx_cdm}
We can adapt \cref{alg:approx} for the CDM model, but we only maintain the approximation
error bounds under special cases of the structure of the ``pulls''. Still, we can use this algorithm
as a principled heuristic and it tends to work well in practice, as we saw in \cref{fig:all_pairs}.

As a first step, we use the alternative parametrization of the model used by \citet[Eq.~(1)]{seshadri2019discovering}, which has fewer parameters.
In this description of the model, utilities and context effects are merged into a single utility-adjusted pull $q_a(z, x) = p_a(z, x) - u_a(x)$, with the special case $q_a(x, x)= 0$. We then have 
 \begin{equation}\label{eq:cdm_alt_param}
   \Pr(a \gets x \mid C) = \frac{\exp({\sum_{w \in C} q_a(w, x)})}{\sum_{y \in C} \exp({\sum_{z \in C} q_a(z, y)})}.
 \end{equation}
 Refer to \citet[Appendix C.1]{seshadri2019discovering} for details of the equivalence between this formulation and the one we use in the main text.

 Matching the notation of the proof of \cref{thm:approx_alg}, we use the shorthand $e_{ax} = \exp({\sum_{w \in C} q_a(w, x)})$.

 To adapt \Cref{alg:approx} to the CDM, we expand $L_i$ to have $nk$ dimensions for each individual-item pair, increasing the runtime to $O((m+kn^2)(1+\lfloor \log_{1+\delta} s \rfloor)^{nk})$. This is only practical if $nk$ is small, but as we have seen, \agree, \disagree, and \promo are all NP-hard even with $n=2$ and $k=2$ or $3$.
 Each individual-item dimension stores $e_{ax}$, the total exp-utility of that item to that individual given that we have included some set of alternatives. When we include an additional item from $\overline C$, we place the new sets in $L_i$ with updated $e_{ax}$ values.

 This only preserves the $\varepsilon$-additive approximation if alternatives (items in $\overline C$) have zero context effects on each other;
 however, they may still have context effects on items in $C$. Formally, we need $q_a(z, z') = 0$ for all $z, z' \in \overline C$ and $a\in A$.
 Although this is a serious restriction, it leaves \agree, \disagree, and \promo NP-hard, as the CDM we used in our proofs had this form
 (see also \Cref{sec:approx_promo} for how to apply \Cref{alg:approx} to \promo).
 If this version of the algorithm is applied to a general CDM, it may experience higher error.
 Nonetheless, our real-data experiments show it to be a good heuristic.

For the following analysis, we assume a CDM with zero context effects between items in $\overline C$. We need to verify that if every item's exp-utility is approximated to within factor $(1+\beta)^{\pm 1}$, then the total disagreement of a set is approximated to within $\varepsilon$ as we had in the MNL case. The approximation error guarantee increases to $4\varepsilon$ in the restricted CDM version---to recover the $\varepsilon$-additive approximation, we can make $\delta$ smaller by a factor of 4 (that is, we could pick $\delta = \varepsilon / (8 k m \binom n 2)$; we instead keep the old $\delta$ for simplicity in the following analysis).

Recall that $Z'$ is the representative in $L_m$ of the optimal set of alternatives $Z^*$. For compactness, we define $T_a$ to be the denominator of \cref{eq:cdm_alt_param}, with $T_a'$ and $T_a^*$ referring to those denominators under the choice sets $C\cup Z'$ and $C\cup Z^*$, respectively. This is where we require zero context effects between alternatives: if alternatives interact, then storing every $e_{ax}$ in the table (from which we can compute $T_a$) is not enough to determine updated choice probabilities when we add a new alternative.

 The difference in the analysis begins when we bound $\Pr(a\gets x \mid C\cup Z')$ on both sides using the fact that each exp-utility sum is approximated within a $1+\beta$ factor (so the probability denomiators $T_a$ are also approximated within this factor): 
\begin{align*}
  \frac{\frac{e_{ax}^*}{1+\beta}}{T_a^* (1+\beta)} &= \frac{1}{(1+\beta)^2}\frac{e_{ax}^*}{T_a^*}\\
  &< \frac{e_{ax}'}{T_a'} = \Pr(a\gets x \mid C\cup Z') \\
  &<\frac{e_{ax}^*(1+\beta)}{ \frac{T_a^*}{1+\beta}}=(1+\beta)^2\frac{e_{ax}^*}{T_a^*}.
\end{align*}
Based on the lower bound, the difference between $\Pr(a\gets x \mid C\cup Z^*)$ and $\Pr(a\gets x \mid C\cup Z')$ could be as large as
\begin{align*}
  \frac{e_{ax}^*}{T_a^*} - \frac{1}{(1+\beta)^2}\frac{e_{ax}^*}{T_a^*} &\le 1-\frac{1}{(1+\beta)^2}.
\end{align*}
Now considering the upper bound, the difference between $\Pr(a\gets x \mid C\cup Z^*)$ and $\Pr(a\gets x \mid C\cup Z')$ could be as large as
\begin{align*}
   (1+\beta)^2\frac{e_{ax}^*}{T_a^*} - \frac{e_{ax}^*}{T_a^*} &\le (1+\beta)^2 - 1.
\end{align*}
Therefore, $|\Pr(a\gets x \mid C\cup Z') - \Pr(b\gets x \mid C\cup Z')|$ can only exceed $|\Pr(a\gets x \mid C\cup Z^*) - \Pr(b\gets x \mid C\cup Z^*)|$ by at most $1-\frac{1}{(1+\beta)^2} + (1+\beta)^2 - 1 = (1+\beta)^2 - \frac{1}{(1+\beta)^2}$.
This is at most $4\beta$:
\begin{align*}
  4 \beta - (1+\beta)^2 + \frac{1}{(1+\beta)^2} &=  \frac{\beta^2(2-\beta^2)}{(1+\beta)^2}\\
  &> 0. \tag{for $0 < \beta < \sqrt{2}$}
\end{align*}

So $D(Z')$ and $D(Z^*)$ are within $4\beta\binom{n}{2}k = 4\varepsilon$.

\subsection{Adapting \Cref{alg:approx} for NL with full guarantees}\label{sec:nl_approx_agree}
We can also adapt \cref{alg:approx} for the NL model, and unlike the CDM,
the $\varepsilon$-additive approximation holds in all parameter regimes.
Recall that the NL tree has two types of leaves: choice set items and alternative items.
Let $P_a$ be the set of internal nodes of individual $a$'s tree that have at least one alternative item as a child and let $p = \max_{a\in A} |P_a|$.
If we know the total exp-utility that alternatives contribute as children of each $v\in P_a$, then we can compute $a$'s choice probabilites over items in $C$ in polynomial time. 

With this in mind, we modify \cref{alg:approx} by having dimensions in $L$ for each individual for each of their nodes in $P_a$.
This results in $\le np$ dimensions. The algorithm then keeps track of the exp-utility sums from alternatives under each node in $P_a$ for each individual.
The exponent in the runtime increases to (at most) $np$, but this remains tractable for some hard instances, such as those in our hardness proofs. In some cases, we can dramatically improve the runtime of the algorithm: if the subtree under an internal node contains only alternatives as leaves in an individuals's tree, then we only need one dimension $L$ for that individual's entire subtree, and it has only two cells: one for sets that contain at least one alternative in that subtree, and one for sets that do not. The only factor that affects the choice probabilities of items in $C$ is whether that subtree is ``active'' and its root can be chosen. 

We now show how the error from exp-utility sums of alternatives propagates to choice probabilities.
In the NL model, $\Pr(a\gets x \mid C)$ is the product of probabilities that $a$ chooses each ancestor of $x$ as $a$ descends down its tree.
Let $v_1, \dots, v_\ell$ be the nodes in $a$'s tree along the path from the root to $x$. For compactness, we use $\Pr(x,Z)$ instead of $\Pr(a\gets x \mid C \cup Z)$ in the following analysis.

Pick $\delta \le ([\varepsilon/(2k \binom n 2)+1]^{1/\ell}-1)/m$ and recall that $\beta = 2m\delta$.
We can use the same analysis as in the proof of \cref{thm:approx_alg} to find that for any set $Z^*\subseteq \overline C$, there exists some $Z'\in L$ such that
\begin{align*}
  \Pr(x, Z^*) &= \Pr(v_1,Z^*)\cdot \dots \cdot \Pr(v_x,Z^*)\\
  &< \Big(\Pr(v_1, Z') + \frac{\beta}{2}\Big) \cdot \dots \cdot \Big(\Pr(v_x, Z')+ \frac{\beta}{2}\Big)\\
  &\le \Pr(x, Z') +\Big(1+\frac{\beta}{2}\Big)^\ell - 1\\
  &\le \Pr(x, Z') + \frac{\varepsilon}{2k\binom n 2}.
\end{align*}

 Now for the lower bound, pick $\delta \le (1-[1-\varepsilon/(2k \binom n 2)]^{1/\ell})/m$. Again from the proof of \cref{thm:approx_alg}:
\begin{align*}
  \Pr(x, Z^*) &= \Pr( v_1, Z^*)\cdot \dots \cdot \Pr( v_x, Z^*)\\
  &> \Big(\Pr(v_1, Z') - \frac{\beta}{2}\Big) \cdot \dots \cdot \Big(\Pr( v_x, Z')- \frac{\beta}{2}\Big)\\
  &\ge \Pr( x,Z' ) +\Big(1-\frac{\beta}{2}\Big)^\ell - 1\\
  &\ge \Pr( x,Z' ) - \frac{\varepsilon}{2k\binom n 2}.
\end{align*}
 Let $h$ be the maximum height of any indivdual's NL tree (so $\ell \le h$). Then, by picking $\delta = \min \{[\varepsilon/(2k \binom n 2)+1]^{1/h}-1, 1-[1-\varepsilon/(2k \binom n 2)]^{1/h}\}/m$, we find that $\Pr(a\gets x \mid C \cup Z^*)$ and $\Pr(a\gets x \mid C \cup Z')$ differ by less than $\varepsilon/(k\binom n 2)$ for all $x\in C$ and $a\in A$, meaning that the total disagreement between $a$ and $b$ cannot differ by more than $\varepsilon$ as before. 

Unfortunately, this means we need to make $\delta$ exponentially (in $h$) smaller in the NL model. Put another way, our error bound gets exponentially worse as $h$ increases if we keep $\delta$ constant. However, we have seen that there are NP-hard families of NL instances in which $h$ is a small constant (e.g., $h = 2$ in our hardness proof), so once again this algorithm is an exponential improvement over brute force. Moreover, the error bound here is often far from tight, since we use the very loose bounds $\Pr(v_i, Z') \le 1$ in the analysis. This means the algorithm will tend to outperform the worst-case guarantee by a significant margin.

\subsection{Adapting \Cref{alg:approx} for \promo}\label{sec:approx_promo}

%
%

\subsubsection{CDM \promo with special case guarantees}
\Cref{alg:approx} can be applied to \promo in the (restricted) CDM model with only a small modification to the CDM version described in \Cref{sec:app_approx_cdm}:
at the end of the algorithm, we return the set that results in the maximum number of individuals having $x^*$ as an $\varepsilon$-favorite item.
Additionally, we choose $\delta = \varepsilon / (10m)$ (we don't need the factors $\binom n 2$ or $k$ since we aren't optimizing $D(Z)$).

Following the analysis in \Cref{sec:app_approx_cdm} (with $\beta = 2m\delta = \varepsilon/5$), we find that $\Pr(a\gets x \mid C\cup Z^*)$ and $\Pr(a\gets x \mid C\cup Z')$ differ by at most $\max\{1-\frac{1}{(1+\varepsilon/5)^2}, (1+\varepsilon/5)^2 - 1\}$ for all $x$.
On the interval $[0, 1]$, this is bounded by $\varepsilon/2$.
Thus, if $x^*$ is the favorite item for $a$ given the optimal choice set $C\cup Z^*$, then it must be an $\varepsilon$-favorite of individual $a$ given $C\cup Z'$ (as always, $Z'$ is the representative of $Z^*$ in $L_m$).
This is because when we go from $C\cup Z^*$ to $C\cup Z'$, the choice probability of $x^*$ can shrink by at most $\varepsilon/2$ and the choice probability for any other item can grow by at most $\varepsilon/2$. Thus, including $Z'$ makes at least as many individuals have $x^*$ as an \emph{$\varepsilon$-favorite} item as including $Z^*$ makes have $x^*$ as a \emph{favorite} item.

This is exactly what it means for \Cref{alg:approx} to $\varepsilon$-approximate \promo in the CDM (when items in $\overline C$ do not exert context effects on each other). Moreover, not having to compute $D(Z)$ makes the runtime of \Cref{alg:approx} $O(m (1+\lfloor \log_{1+\delta} s \rfloor)^{nk})$ when applied to \promo in the CDM. In the general CDM, this algorithm is only a heuristic.

\subsubsection{NL \promo with full guarantees}
A very similar idea allows us to apply the NL version of \Cref{alg:approx} from \Cref{sec:nl_approx_agree} to \promo and retain an approximation guarantee. As before, use the NL version and return the set that results in the maximum number of individuals having $x^*$ as an $\varepsilon$-favorite item. However, we instead use $\delta = \min \{(\varepsilon/4+1)^{1/h}-1, 1-(1-\varepsilon/4)^{1/h}\}/m$, which by the analysis in \Cref{sec:nl_approx_agree} results in $\Pr(a\gets x \mid C\cup Z^*)$ and $\Pr(a\gets x \mid C\cup Z')$ differing by at most $\varepsilon/2$. As in the CDM case, this guarantees that if $x^*$ is the favorite item for $a$ given the optimal choice set $C\cup Z^*$, then it must be an $\varepsilon$-favorite of $a$ given $C\cup Z'$. Therefore this version of \Cref{alg:approx} $\varepsilon$-approximates \promo in the NL model with runtime $O(m (1+\lfloor \log_{1+\delta} s \rfloor)^{np})$.

\section{Additional experiment details}
\subsection{Simple example of poor performance for Greedy}
As we saw in experimental data, Greedy can perform poorly even in small instances of \agree.
Below we provide an MNL instance with $n=m=k=2$ for which the error of the greedy solution is approximately 1.
With only two individuals, $0 \le D(Z) \le 2$, so an error of 1 is very large.

In the bad instance for greedy, $A=\{a, b\}$, $C= \{x, y\}$, $\overline C = \{p, q\}$, and the utilities are as follows. 
\begin{equation*}
  \begin{aligned}[c]
    u_a(x) &= 8\\
    u_a(y) &= 2\\
    u_a(p) &=10\\
    u_a(q) &= 0
  \end{aligned}
  \qquad
  \begin{aligned}[c]
  u_b(x) &= 8\\
    u_b(y) &= 8\\
    u_b(p) &=0\\
    u_b(q) &= 15
  \end{aligned}
\end{equation*}
In this instance of \agree, the greedy solution is $D(\emptyset) \approx 0.9951$ (including either $p$ or $q$ alone increases disagreement), while the optimal solution is $D(\{p, q\}) \approx 0.0009$.

\subsection{All-pairs agreement results for MIBLP}
\label{sec:MIBLP_results}
\Cref{fig:all_pairs_MIBLP} shows the comparison in performance between \Cref{alg:approx} and the MIBLP approach for the all-pairs \agree and \disagree experiment. The methods perform nearly identically on both \allstate and \yoochoose. The MIBLP approach performs marginally better in some cases of \yoochoose \agree. As noted in the paper, the MIBLP heuristic is considerably faster for the values of $\varepsilon$ we used ($12$x and $240$x on \yoochoose and \allstate, respectively; speed differences vary significantly depending on $\varepsilon$), but provides no a priori performance guarantee and cannot be applied to CDM or NL. Nonetheless, we can see that it performs very competitively and would be a good approach to use in practice for MNL \agree and \disagree. 
\begin{figure}[h]
    \centering
    \includegraphics[width=0.7\columnwidth]{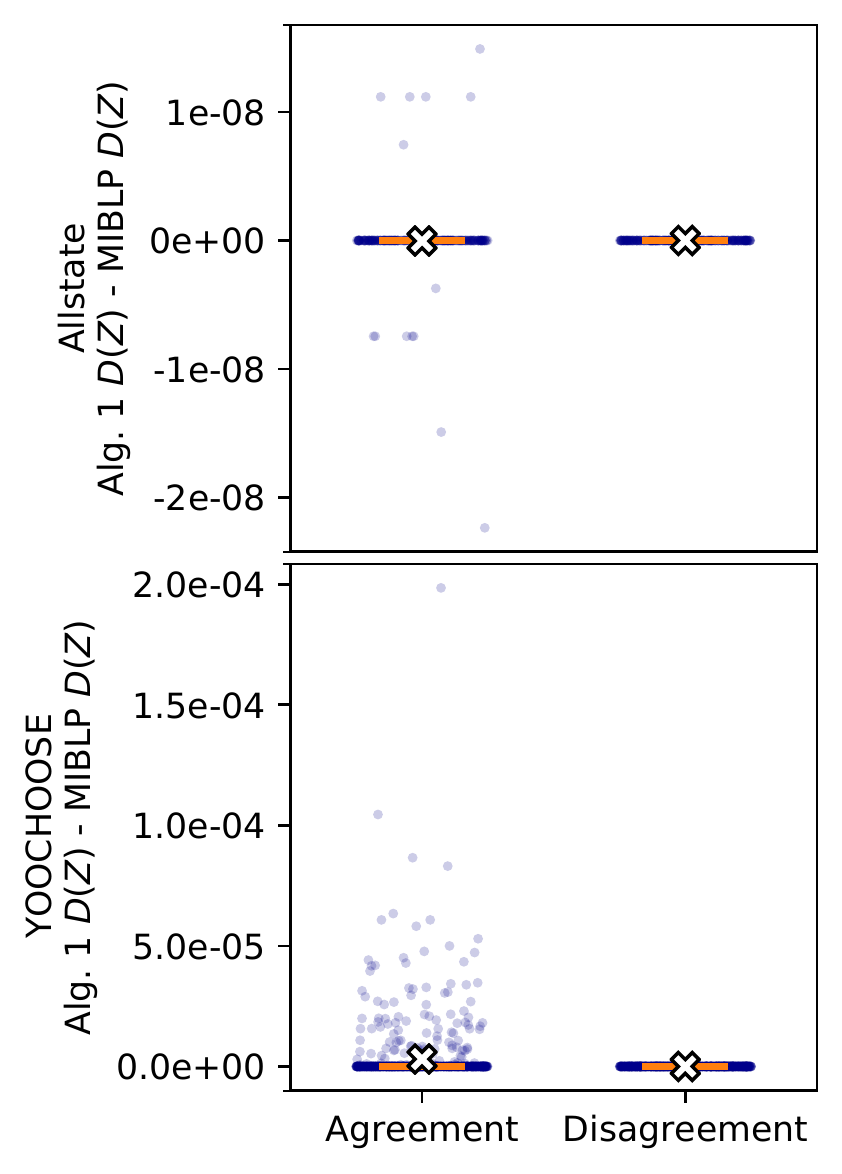}
    \caption{MIBLP vs.~\Cref{alg:approx} performance box plots
    when applied to all 2-item choice sets in \allstate and \yoochoose under MNL. Each point is the difference in $D(Z)$ when MIBLP and \Cref{alg:approx} are run on a choice set, and Xs mark means. }
    \label{fig:all_pairs_MIBLP}
\end{figure}

\begin{figure}[h!]
    \centering
    \includegraphics[width=\columnwidth]{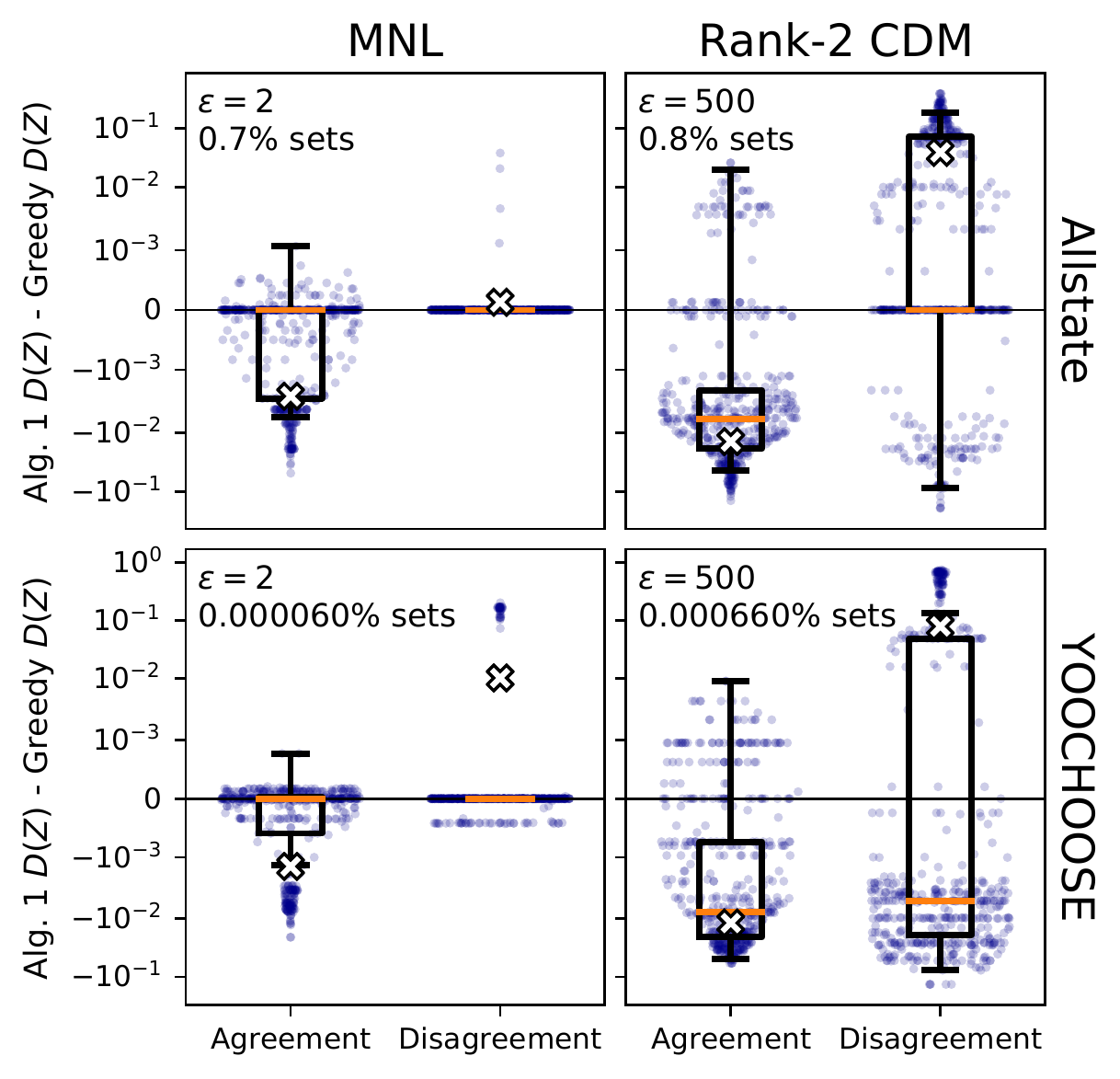}
    \caption{Results of the agreement experiment with 500 choice sets sampled uniformly from each dataset. Compare with \cref{fig:all_pairs} in the main text. Again, \Cref{alg:approx} has better mean performance in all cases. The larger values of $\varepsilon$ result in slightly worse performance on the margins than in \cref{fig:all_pairs}, but also fewer sets computed.}
    \label{fig:sampled_choice_sets}
\end{figure}

\subsection{Choice sets sampled from data}
\label{sec:sampled_choice_sets}
We repeated the all-pairs agreement experiment with 500 choice sets of size up to 5 sampled uniformly from each dataset, allowing us to evaluate the performance of \Cref{alg:approx} on realistic choice sets. We limited the size of sampled choice sets since the CDM version of \Cref{alg:approx} scales poorly with $|C|$ (see \cref{sec:app_approx_cdm}). For this version of the experiment, we fixed larger values of $\varepsilon$ (2 for MNL, 500 for CDM) to handle larger choice sets and to keep running time down. Again, \Cref{alg:approx} has better mean performance in every case (\cref{fig:sampled_choice_sets}), showing that it performs well on real choice sets.

\section{A note on ethical considerations}
\label{sec:ethics}

Influencing the preferences of decision-makers has the potential for malicious applications, so it is important to address the ethical context of this work.

Any problem with positive social applications (e.g., \agree: encouraging consensus, \promo: promoting environmentally-friendly transportation options, \disagree: increasing diversity of opinions) has the potential to be used for ill. This should not prevent us from seeking methods to acheive these positive ends, but we should certainly be cognizant of the possibility of unintended applications. In a different vein, understanding when a group is susceptible to undesired interventions (or detecting such interventions) makes problems like \disagree worth studying from an adversarial perspective. Along these lines, our hardness results are encouraging since optimal malicious interventions are difficult. 

Finally, we note that all of the theoretical problems we study presuppose access to choice data from which preferences can be learned and the ability to influence choice sets. Any entity which has both of these (such as an online retailer, a government deciding transportation policy, etc.) already has significant power to influence choosers. If such an entity had malicious intent, then near-optimal \disagree solutions would be the least of our concerns. 

To summarize, these problems are worth studying because of (1) their purely theoretical value in furthering the field of discrete choice, (2) their potential for positive applications, (3) insight into the potential for harmful manipulation by an adversary, and (4) the minimal additional risk from undesired use of our methods.

\end{document}